\documentclass[a4paper,UKenglish, autoref, thm-restate, cleveref]{lipics-v2021}

\hideLIPIcs
\nolinenumbers 

\graphicspath{{./figures/}}
\title{Locally Correct Interleavings between Merge Trees}
\author{Thijs Beurskens}{Department of Mathematics and Computer Science, TU Eindhoven, The Netherlands}{t.p.j.beurskens@tue.nl}{https://orcid.org/0009-0004-1185-1545}{Supported by the Dutch Research Council (NWO) under project no. OCENW.M20.089.}

\author{Tim Ophelders}{Department of Information and Computing Science, Utrecht University, The Netherlands \and Department of Mathematics and Computer Science, TU Eindhoven, The Netherlands}{t.a.e.ophelders@uu.nl}{https://orcid.org/0000-0002-9570-024X}{Supported by the Dutch Research Council (NWO) under project no. VI.Veni.212.260.}

\author{Bettina Speckmann}{Department of Mathematics and Computer Science, TU Eindhoven, The Netherlands}{b.speckmann@tue.nl}{https://orcid.org/0000-0002-8514-7858}{}

\author{Kevin Verbeek}{Department of Mathematics and Computer Science, TU Eindhoven, The Netherlands}{k.a.b.verbeek@tue.nl}{https://orcid.org/0000-0003-3052-4844}{}
\authorrunning{T. Beurskens, T. Ophelders, B. Speckmann, K. Verbeek} 

\keywords{Interleaving distance, merge trees, local correctness, matchings} 
\ccsdesc[500]{Theory of computation~Computational geometry}

\acknowledgements{Special thanks to Willem Sonke for helpful discussions on the topic of this paper and creating Figure~\ref{fig:merge-tree-terrain}.} 
\nolinenumbers 

\EventEditors{John Q. Open and Joan R. Access}
\EventNoEds{2}
\EventLongTitle{42nd Conference on Very Important Topics (CVIT 2016)}
\EventShortTitle{CVIT 2016}
\EventAcronym{CVIT}
\EventYear{2016}
\EventDate{December 24--27, 2016}
\EventLocation{Little Whinging, United Kingdom}
\EventLogo{}
\SeriesVolume{42}
\ArticleNo{23}

\usepackage{amsmath, amssymb, amsfonts}
\usepackage{mathtools}
\usepackage{multicol}
\usepackage{apptools}
\AtAppendix{\counterwithin{theorem}{section}}
\AtAppendix{\counterwithin{lemma}{section}}
\AtAppendix{\counterwithin{observation}{section}}

\newcommand{\etal}{\emph{et al.}}

\newcommand{\bigO}{\mathcal{O}}

\newcommand{\dom}{\mathrm{dom}}
\newcommand{\an}[2]{#1|^{#2}}

\newcommand{\shift}[1]{\sigma_{#1}}
\newcommand{\intdist}[1]{d^{{#1}}}
\newcommand{\resdist}[1]{d_{#1}}

\newcommand{\eps}{\varepsilon}
\renewcommand{\phi}{\varphi}

\newcommand{\I}{\mathcal{I}}
\renewcommand{\P}{\mathcal{P}}
\newcommand{\Q}{\mathcal{Q}}
\newcommand{\R}{\mathcal{R}}

\renewcommand{\mid}{:}

\newcommand{\enumit}[1]{\textcolor{darkgray}{\sffamily\bfseries\upshape\mathversion{bold}#1}}

\begin{document}
\maketitle

\begin{abstract}
Temporal sequences of terrains arise in various application areas. To analyze them efficiently, one generally needs a suitable abstraction of the data as well as a method to compare and match them over time. In this paper we consider merge trees as a topological descriptor for terrains and the interleaving distance as a method to match and compare them. An \emph{interleaving} between two merge trees consists of two maps, one in each direction. These maps must satisfy ancestor relations and hence introduce a ``shift'' between points and their image. An optimal interleaving minimizes the maximum shift; the interleaving distance is the value of this shift. However, to study the evolution of merge trees over time, we need not only a number but also a meaningful matching between the two trees. The two maps of an optimal interleaving induce a matching, but due to the bottleneck nature of the interleaving distance, this matching fails to capture local similarities between the trees. In this paper we hence propose a notion of local optimality for interleavings. To do so, we define the \emph{residual interleaving distance}, a generalization of the interleaving distance that allows additional constraints on the maps. This allows us to define \emph{locally correct interleavings}, which use a range of shifts across the two merge trees that reflect the local similarity well. We give a constructive proof that a locally correct interleaving always exists.
\end{abstract}

\section{Introduction}\label{sec:introduction}
Terrains that vary over time are a common data type in various application areas, such as scientific computing or geographic information science. Typically such data sets are very large; to analyze them efficiently one often needs a compact abstraction that captures salient features. Topological data analysis (TDA) provides several \emph{topological descriptors}~\cite{yan2021scalar}, such as persistence diagrams, Reeb graphs, and Morse-Smale complexes, that can serve as abstractions for 2D, 3D, or higher-dimensional terrains.
In this paper we focus in particular on \emph{merge trees}, which are graph-based descriptors that encode the evolution of connected components of sublevel or superlevel sets of a terrain (see Figure~\ref{fig:merge-tree-terrain}).

\begin{figure}
    \centering
    \includegraphics{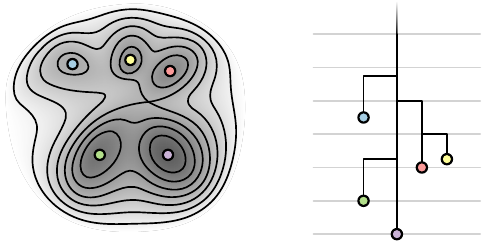}
    \caption{A terrain and its merge tree; leaves and internal vertices represent minima and saddles.}
    \label{fig:merge-tree-terrain}
\end{figure}

To analyze temporal sequences of merge trees one does not only need a compact representation, but also efficient ways to compare and match them over time. For merge trees in particular, several distance measures have been proposed, including the edit distance~\cite{sridharamurthy2020edit}, the Wasserstein distance~\cite{pont2022wasserstein}, and the interleaving distance~\cite{morozov2013interleaving}. However, a distance measure returns only a numeric value that quantifies the similarity of two merge trees~$T_1$ and~$T_2$. In many application scenarios, we would also like to track the merge trees over time and hence need to understand which parts of $T_1$ most closely correspond to parts of $T_2$. That is, we are interested in a matching between $T_1$ and~$T_2$.

The interleaving distance in fact relies on a type of matching between the two trees, the so-called \emph{interleaving}.
An interleaving is a pair of ancestor-preserving maps, one from~$T_1$ to~$T_2$ and one from~$T_2$ to~$T_1$, whose compositions send points to ancestors: if a point~$x$ of~$T_1$ is mapped to a point~$y$ of~$T_2$, then~$y$ must be mapped to an ancestor of~$x$, and vice versa.
The \emph{interleaving distance} is the maximum \emph{shift}---the distance between a point and its image---minimized over all interleavings between the trees. We give exact definitions of the interleaving distance and interleavings in Sections~\ref{sec:preliminaries} and~\ref{sec:generalizing-interleaving-distance}.

\begin{figure}[b]
    \centering
    \includegraphics{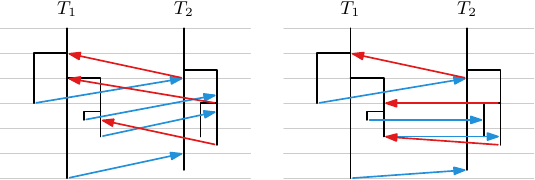}
    \caption{Two optimal interleavings between merge trees~$T_1$ and~$T_2$. Intuitively, the right interleaving induces a tighter and more meaningful matching than the left interleaving.}
    \label{fig:good-vs-bad-optimal-interleavings}
\end{figure}

Two merge trees typically admit many optimal interleavings.
The interleaving distance is a bottleneck measure and hence a given optimal interleaving might be locally quite ``loose'': the optimal shift might exceed what is needed and thereby fail to capture local similarities between the trees; see Figure~\ref{fig:good-vs-bad-optimal-interleavings}. Inspired by the concept of locally correct Fréchet matchings for curves~\cite{buchin2019locally}, in the following we are hence developing a notion of locally optimal interleavings which allows us to distinguish between ``tight'' and ``loose'' optimal interleavings. 

\subparagraph{Contributions.}
A locally correct Fréchet matching is optimal even when considering any two matched subcurves.
This is a desirable feature, however, it cannot be directly reproduced using interleavings.
First of all, there might not even be a one-to-one correspondence between the points of~$T_1$ and~$T_2$. 
Furthermore, two optimal interleavings between different pairs of subtrees do not necessarily combine into an optimal interleaving between the two complete trees. 
As a first step we hence introduce a novel generalization of the interleaving distance that can handle additional constraints on the maps. 
Specifically, we aim to use constraints that force a point~$x$ to be mapped to an ancestor of a particular point~$y$ within the other tree.
Over interleavings that satisfy such constraints, we aim to minimize the maximum shift over points~$x$ that are either unconstrained or map to a strict ancestor of their corresponding point~$y$.
We refer to the resulting distance as the \emph{residual interleaving distance}.
In Section~\ref{sec:generalizing-interleaving-distance}, we provide formal definitions and show that an optimal interleaving exists for any constraints.

In Section~\ref{sec:locally-correct-interleavings} we then propose a definition for \emph{locally correct interleavings} that builds upon the residual interleaving distance. 
Specifically, we define a restriction of an interleaving~${\I = (\alpha, \beta)}$ to two subsets~$S_1$ of~$T_1$ and~$S_2$ of~$T_2$ as the restriction of~$\alpha$ to~$S_1$ together with the restriction of~$\beta$ to~$S_2$.
We call $\I$ locally correct if it minimizes the residual interleaving distance with respect to any restriction of~$\I$.
We constructively prove that a locally correct interleaving always exists.

\subparagraph{Related work.}
The interleaving distance between merge trees was defined by Morozov \etal~\cite{morozov2013interleaving} as a tool to study the stability of merge trees.
Agarwal \etal~\cite{agarwal2018computing} establish a connection between the interleaving distance and the Gromov-Hausdorff distance, and they argue that approximating the interleaving distance within a factor of~$3$ is NP-hard.
They also describe an~$\bigO(\min(n, \sqrt{rn}))$-approximation algorithm, where~$n$ is the total size of the trees and~$r$ is the ratio of the longest to the shortest edge.
Touli and Wang~\cite{touli2022fpt} provide an equivalent definition of the interleaving distance in terms of a single map and use this characterization to design an FPT-algorithm for computing the interleaving distance exactly.
More recently, Gasparovic \etal~\cite{gasparovic2024intrinsic} show that the interleaving distance can be formulated in terms of a related metric~\cite{munch2019the} for labeled merge trees, also known as phylogenetic trees.

Several attempts have been made to develop versions of the interleaving distance that can be used in practical applications.
In particular, the work by Gasparovic \etal~\cite{gasparovic2024intrinsic} gave rise to heuristic algorithms for computing the interleaving distance~\cite{curry2022decorated} and for comparing terrains~\cite{yan2022geometry, yan2020structural}.
Pegoraro~\cite{pegoraro2025graph-matching} uses yet another reformulation of the interleaving distance, in terms of couplings, as the basis of a heuristic for computing the distance via linear integer programming.
Lastly, Beurskens \etal~\cite{beurskens2025relating} introduce a variant of the interleaving distance that imposes total orders on the merge trees, which can be computed efficiently via a connection to the Fréchet distance between one-dimensional curves.

\section{Preliminaries}\label{sec:preliminaries}
A \emph{merge tree} is a pair~$(T, f)$, where~$T$ is a finite rooted tree, and~$f \colon T \to \mathbb{R} \cup \{\infty\}$ is a continuous \emph{height function} on $T$ that is strictly increasing towards the root and such that the root has height~$\infty$.
We identify~$T$ with its \emph{topological realization}: the topological space where each edge of~$T$ is represented by a unit segment and connections are according to the adjacencies of~$T$.
We distinguish elements of the topological realization $T$, which we refer to as \emph{points} of~$T$, from vertices of the combinatorial tree, which we denote by~$V(T)$.
When clear from context, we simply write~$T$ to refer to the pair~$(T, f)$.

For two points~$x_1, x_2$ of~$T$, we say~$x_1$ is a \emph{descendant} of~$x_2$, written $x_1 \preceq x_2$, if there exists an~$f$-monotonically increasing path in~$T$ from~$x_1$ to~$x_2$.
If furthermore~$x_1 \neq x_2$ then~$x_1$ is a \emph{strict descendant} of~$x_2$.
We say~$x_2$ is a \emph{(strict) ancestor} of~$x_1$ if~$x_1$ is a (strict) descendant of~$x_2$.
Lastly, for a point~$x$ and a height value~$h \ge f(x)$, we use~$\an{x}{h}$ to denote the unique ancestor of~$x$ at height~$f(\an{x}{h}) = h$ (see Figure~\ref{fig:merge-tree-interleaving} left).

\begin{figure}
	\centering
	\includegraphics{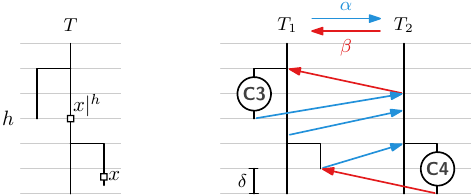}
	\caption{Left: a merge tree~$T$; the point~$\an{x}{h}$ is the ancestor of~$x$ at height~$h$. Right: parts of a pair of~$\delta$-compatible maps.}
	\label{fig:merge-tree-interleaving}
\end{figure}

\subparagraph{Compatible maps.}
The original definition of the interleaving distance between two merge trees~$(T_1, f_1)$ and~$(T_2, f_2)$ was given by Morozov \etal~\cite{morozov2013interleaving}.
\begin{definition}[Morozov \etal~\cite{morozov2013interleaving}]\label{def:interleaving-distance}
    Fix a value~$\delta \ge 0$.
	Two continuous maps~$\alpha \colon T_1 \to T_2$ and~$\beta \colon T_2 \to T_1$ are~$\delta$-\emph{compatible} between~$T_1$ and~$T_2$ if for all~$x \in T_1$ and all~$y \in T_2$:
	
	\begin{multicols}{2}
		\begin{enumerate}[({C}1)]
			\item $f_2(\alpha(x)) = f_1(x) + \delta$,
			\item $f_1(\beta(y)) = f_2(y) + \delta$,
			\item $\beta(\alpha(x)) = \an{x}{f_1(x) + 2\delta}$, and
			\item $\alpha(\beta(y)) = \an{y}{f_2(y) + 2\delta}$.
		\end{enumerate}
	\end{multicols}
	
	\noindent
	The \emph{interleaving distance}~$\intdist{}(T_1, T_2)$ between~$T_1$ and~$T_2$ is the infimum~$\delta$ for which there exist~$\delta$-compatible maps (see Figure~\ref{fig:merge-tree-interleaving} right).
\end{definition}

\subparagraph{Critical values.}
Both Agarwal \etal~\cite{agarwal2018computing} and Touli and Wang~\cite{touli2022fpt} have shown that the interleaving distance between any two merge trees attains a value in a set of \emph{critical values}.
Specifically, let~$\Delta \coloneqq \Delta(T_1, T_2)$ be the union of three sets~$\Delta_1$, $\Delta_2$, and~$\Delta_3$, with
\begin{align}\label{eq:critical-values}
	\Delta_1 &\coloneqq \{|f_1(v) - f_2(w)| \mid v \in V(T_1), w \in V(T_2)\},\nonumber\\
	\Delta_2 &\coloneqq \{\tfrac{1}{2} |f_1(v_1) - f_1(v_2)| \mid v_1, v_2 \in V(T_1)\},\\ 
	\Delta_3 &\coloneqq \{\tfrac{1}{2} |f_2(w_1) - f_2(w_2)| \mid w_1, w_2 \in V(T_2)\}.\nonumber
\end{align}

\begin{lemma}[\cite{agarwal2018computing, touli2022fpt}]\label{lem:critical-values}
	For any two merge trees~$T_1$ and~$T_2$, it holds that $\intdist{}(T_1, T_2) \in \Delta(T_1, T_2)$.
\end{lemma}

\section{Generalizing the Interleaving Distance}\label{sec:generalizing-interleaving-distance}
The interleaving distance is based on a pair of maps.
As noted by Agarwal \etal~\cite{agarwal2018computing}, we can relax the requirements on these maps to allow varying height differences between points and their images.
Building on this relaxed definition, we introduce a generalized version of the interleaving distance that accommodates constraints on the maps we consider.
Specifically, each constraint specifies that a point~$x$ must be mapped to an ancestor of a particular point~$y$ within the other tree.
Our goal is then to find the optimal pair of maps that ``extends'' a given set of constraints.
We call the corresponding distance the \emph{residual~interleaving distance}.
To define it formally, we first introduce \emph{partial up-maps}, which represent the allowed maps between the two merge trees.
In Section~\ref{sec:complete-interleavings} we relax Definition~\ref{def:interleaving-distance} using partial up-maps.
In Section~\ref{sec:residual-interleaving-distance} we then give the formal definition of the residual interleaving distance.

\subparagraph{Partial up-maps.}
Let~$(T_1, f_1)$ and~$(T_2, f_2)$ be two merge trees.
An \emph{arrow} from~$T_1$ to~$T_2$ is a pair of points~$(x, y) \in T_1 \times T_2$ such that~$f_1(x) \le f_2(y)$.
The \emph{shift} of an arrow~$(x, y)$ is the height difference between~$x$ and~$y$, denoted~$\shift{}(x, y) \coloneqq f_2(y) - f_1(x)$.
Consider a subset of points~$S \subseteq T_1$.
A map~$\phi \colon S \to T_2$ is a \emph{partial up-map} from~$S$ to~$T_2$ if (i)~$f_2(\phi(x)) \ge f_1(x)$ for all points~$x$ of~$S$ and (ii) it \emph{preserves ancestors}, that is, $x_1 \preceq x_2$ implies~$\phi(x_1) \preceq \phi(x_2)$ for all points~$x_1, x_2$ of~$S$.
We use~$\dom(\phi) \coloneqq S$ to denote the \emph{domain} of~$\phi$.
For each point~$x$ of~$S$, the pair~$(x, \phi(x))$ is an arrow.
We denote the set of arrows of~$\phi$ by~$A_\phi \coloneqq \{(x, \phi(x)) \mid x \in S\}$.
We define the \emph{shift} of the partial up-map~$\phi$ as the supremum shift over all arrows of~$\phi$:
\begin{equation}\label{eq:shift}
    \shift{}(\phi) \coloneqq \sup\{\shift{}(a) \mid a \in A_\phi\}.
\end{equation}

Given an arrow~$(x, y)$ from~$S$ to~$T_2$, we say~$\phi$ \emph{extends}~$(x, y)$ if~$\phi(x) \succeq y$, and we say~$\phi$ \emph{uses}~$(x, y)$ if~$\phi(x) = y$.
Similarly, for a subset~$S' \subseteq S$ and a partial up-map~$\phi'$ from~$S'$ to~$T_2$, we say~$\phi$ \emph{extends} or \emph{uses}~$\phi'$ if, respectively,~$\phi$ extends or uses all arrows of~$\phi'$ (Figure~\ref{fig:partial-maps} left).

For any~$\delta \ge 0$, we can define a partial up-map~$\phi^\uparrow[\delta]$ from~$S$ to~$T_2$ that extends~$\phi$ and such that the shift of each arrow is at least~$\delta$.
Specifically, for each point~$x$ of~$S$, let~$h_x \coloneqq \max(f_1(x) +\delta, f_2(\phi(x)))$ and define~$\phi^\uparrow[\delta](x) \coloneqq \an{\phi(x)}{h_x}$ (see Figure~\ref{fig:partial-maps} middle).

\begin{lemma}\label{lem:lifted}
	For any~$\delta \ge 0$, the map~$\phi^\uparrow[\delta]$ is a partial up-map that extends~$\phi$.
\end{lemma}
\begin{proof}
    Let~$\phi' \coloneqq \phi^\uparrow[\delta]$.
    By construction, we have~$f_2(\phi'(x)) \ge f_1(x)$.
	To see that~$\phi'$ preserves ancestors, take two points~$x_1, x_2$ of~$S$ such that~$x_2$ is an ancestor of~$x_1$.
	Since~$\phi$ preserves ancestors,~$\phi(x_2)$ is an ancestor of~$\phi(x_1)$.
	If~$\phi'(x_1) = \phi(x_1)$, we get~${\phi'(x_1) \preceq \phi(x_2) \preceq \phi'(x_2)}$.
	Otherwise~$\phi'(x_1)$ is a strict ancestor of~$\phi(x_1)$.
	By construction,~$f_2(\phi'(x_1)) = f_1(x_1) + \delta \le f_1(x_2) + \delta \le f_2(\phi'(x_2))$.
	Both~$\phi'(x_1)$ and~$\phi'(x_2)$ are ancestors of~$\phi(x_1)$, so~$\phi'(x_2)$ is an ancestor of~$\phi'(x_1)$.
	So~$\phi'$ is a partial up-map, and by construction also extends~$\phi$.
\end{proof}

\noindent
We define arrows and up-maps from a subset~$S$ of~$T_2$ to~$T_1$ symmetrically.

\subparagraph{Partial interleavings.}
Fix two subsets~$S_1 \subseteq T_1$ and~$S_2 \subseteq T_2$.
A \emph{partial interleaving}~$\P$ between~$S_1$ and~$S_2$ is a pair of partial up-maps~$\phi \colon S_1 \to T_2$ and~$\psi \colon S_2 \to T_1$ such that for all points~$x$ of~$S_1$ and all points~$y$ of~$S_2$ it holds that if~$\phi(x) \preceq y$ then~$\psi(y) \succeq x$, and if~$\psi(y) \preceq x$ then~$\phi(x) \succeq y$ (Figure~\ref{fig:partial-maps} right).
The \emph{shift} of~$\P$, denoted~$\shift{}(\P)$, is the maximum of the shift of~$\phi$ and the shift of~$\psi$.
We define the arrows of~$\P$ as the union of the arrows of~$\phi$ and the arrows of~$\psi$, denoted~$A_\P \coloneqq A_\phi \cup A_\psi$.
Given another partial interleaving~$\P' = (\phi', \psi')$, we say~$\P$ \emph{extends}~$\P'$ if~$\phi$ extends~$\phi'$ and~$\psi$ extends~$\psi'$.
Similarly, we say $\P$ \emph{uses}~$\P'$ if $\phi$ uses $\phi'$ and $\psi$ uses $\psi'$.
For any~$\delta \ge 0$, the pair~$(\phi^\uparrow[\delta], \psi^\uparrow[\delta])$ is a partial interleaving.

\begin{figure}
    \centering
    \includegraphics{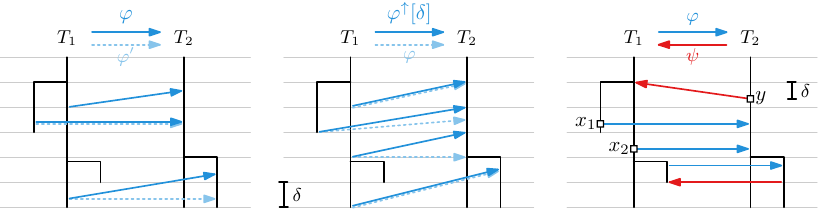}
    \caption{Left: a partial up-map $\phi$ that extends another partial up-map $\phi'$. Middle: a partial up-map~$\phi$ and the corresponding map~$\phi^\uparrow[\delta]$; each non-dashed arrow has shift at least~$\delta$. Right: a partial interleaving with shift~$\delta$; note that~$\psi(y)$ is an ancestor of both~$x_1$ and~$x_2$.}
    \label{fig:partial-maps}
\end{figure}

\begin{lemma}
	For any~$\delta \ge 0$, the pair~$(\phi^\uparrow[\delta], \psi^\uparrow[\delta])$ is a partial interleaving that extends~$(\phi, \psi)$.
\end{lemma}
\begin{proof}
    Denote~$\phi' \coloneqq \phi^\uparrow[\delta]$ and~$\psi' \coloneqq \psi^\uparrow[\delta]$.
    By Lemma~\ref{lem:lifted}, the maps~$\phi'$ and~$\psi'$ are partial up-maps that extend~$\phi$ and~$\psi$, respectively.
    Therefore, it suffices to show that the pair~$(\phi',\psi')$ is a partial interleaving.
	Let~$x$ be a point of~$S_1$ and~$y$ be a point of~$S_2$ such that~$\phi'(x) \preceq y$.
	By construction, the point~$\phi'(x)$ is an ancestor of~$\phi(x)$, so $y$ is also an ancestor of~$\phi(x)$.
	Since~$(\phi, \psi)$ is a partial interleaving, the point~$\psi(y)$ must be ancestor of~$x$.
	As~$\psi'(y)$ is an ancestor of~$\psi(y)$, we get~$\psi'(y) \succeq x$.
	Symmetrically, if~$\psi'(y) \preceq x$ then~$\phi'(x) \succeq y$.
\end{proof}

\noindent
For a partial interleaving~$\P = (\phi,\psi)$, we use~$\P.\phi$ and~$\P.\psi$ to denote~$\phi$ and~$\psi$, respectively.

\subsection{Complete Interleavings}\label{sec:complete-interleavings}
If $S_1=T_1$, we call a partial up-map $\alpha\colon S_1\to T_2$ a \emph{complete up-map}, or simply an \emph{up-map}, from $T_1$ to $T_2$.
Symmetrically, if $S_2=T_2$, we call a partial up-map $\beta\colon S_2\to T_1$ a (complete) up-map from $T_2$ to $T_1$.
Correspondingly, if $S_1=T_1$ and $S_2=T_2$, we call the partial interleaving $(\alpha,\beta)$ a \emph{(complete) interleaving}.
The defining property for a complete interleaving then becomes that for all~$x \in T_1$ it holds that~$\beta(\alpha(x)) \succeq x$ and for all~$y \in T_2$ it holds that~$\alpha(\beta(y)) \succeq x$.
In Lemma~\ref{lem:interleaving-equals-compatible}, we show that the interleaving distance is equal to the infimum~$\delta$ for which~$T_1$ and~$T_2$ admit an interleaving with shift~$\delta$.
First, it is not hard to see that any pair of~$\delta$-compatible maps form an interleaving with shift~$\delta$.
\begin{observation}\label{obs:compatible-interleaving}
    Any pair of~$\delta$-compatible maps is an interleaving with shift~$\delta$.
\end{observation}

\noindent
Conversely, for a complete interleaving~$(\alpha, \beta)$ with shift~$\delta$, we show that there always exists a pair of~$\delta$-compatible maps.
In particular, the pair~$(\alpha^\uparrow[\delta],\beta^\uparrow[\delta])$ is~$\delta$-compatible.

\begin{lemma}\label{lem:interleaving-compatible}
For any interleaving~$(\alpha, \beta)$ with shift~$\delta$,~$\alpha^\uparrow[\delta]$ and~$\beta^\uparrow[\delta]$ are~$\delta$-compatible maps.
\end{lemma}
\begin{proof}
    Denote~$\alpha' \coloneqq \alpha^\uparrow[\delta]$ and~$\beta' \coloneqq \beta^\uparrow[\delta]$.
	We first show that~$\alpha'$ is continuous; proving that~$\beta'$ is continuous is analogous.
	Let~$x$ be a point of~$T_1$ and consider a small neighborhood~$X$ of~$x$ in which each point~$x'$ is reachable from~$x$ by an~$f_1$-monotone path.
	To prove continuity of~$\alpha'$ at~$x$, we show that there exists an~$f_2$-monotone path between~$\alpha'(x)$ and~$\alpha'(x')$ and that the height difference between~$\alpha'(x)$ and~$\alpha'(x')$ is bounded.
    By construction and choice of~$\delta$, it follows that~$\alpha'$ maps each point of~$T_1$ to a point of~$T_2$ exactly~$\delta$ higher.
    In particular, we obtain~$|f_2(\alpha'(x')) - f_2(\alpha'(x))| = |f_1(x') - f_1(x)|$.
	As~$\alpha'$ is an up-map, we know that~$x \preceq x'$ implies~$\alpha'(x) \preceq \alpha'(x')$, and that~$x' \preceq x$ implies~$\alpha'(x') \preceq \alpha'(x)$.
    In other words, there exists an~$f_2$-monotone path between~$\alpha'(x)$ and~$\alpha'(x')$.
	So~$\alpha'$ is continuous.
	
	Lastly, as~$(\alpha', \beta')$ is an interleaving, and the shift of each arrow is exactly~$\delta$, it follows that~$\beta'(\alpha'(x)) = \an{x}{f_1(x) + 2\delta}$ for all~$x \in T_1$, and~$\alpha'(\beta'(y)) = \an{y}{f_2(y) + 2\delta}$ for all~$y \in T_2$.
\end{proof}

\noindent
Combining Observation~\ref{obs:compatible-interleaving} and Lemma~\ref{lem:interleaving-compatible} we obtain Lemma~\ref{lem:interleaving-equals-compatible}.

\begin{restatable}{lemma}{interleavingdistanceinterleavings}
\label{lem:interleaving-equals-compatible}
	$\intdist{}(T_1, T_2) = \inf\{\delta \mid T_1 \text{ and } T_2 \text{ admit a complete interleaving with shift } \delta \}$.
\end{restatable}

\subparagraph{Optimal interleavings.}
The interleaving distance is defined as an infimum, but both Gasparovic \etal~\cite{gasparovic2019intrinsic} and Pegoraro~\cite{pegoraro2025graph-matching} proved that the infimum can be replaced by a minimum: there always exists an interleaving with shift~$\intdist{}(T_1, T_2)$ (Theorem~\ref{thm:existence-optimal-interleaving}).
We call an interleaving that realizes the interleaving distance an \emph{optimal interleaving}.
Both proofs rely on alternative definitions of the interleaving distance, namely in terms of labellings~\cite{gasparovic2019intrinsic} and in terms of couplings~\cite{pegoraro2025graph-matching}.
We now give a more direct proof of Theorem~\ref{thm:existence-optimal-interleaving}.
To do so, we first show that for an interleaving whose shift is not equal to a critical value, there exists an interleaving whose shift is strictly smaller.

\newpage
\begin{lemma}\label{lem:existence-optimal-interleaving}
	For any complete interleaving~$\I = (\alpha, \beta)$ with shift~$\delta \notin \Delta$, there exists another interleaving~$\I'$ whose shift is equal to a critical value that is strictly less than~$\delta$.
\end{lemma}
\begin{proof}
    Let~$\delta'$ be the greatest value in~$\Delta(T_1, T_2)$ that satisfies~$\delta' \le \delta$.
    As there are only finitely many critical values, there is such a value~$\delta'$.
    W.l.o.g., we assume that all arrows of~$\I$ have shift at least~$\delta'$.
    Otherwise, we can replace~$\alpha$ and~$\beta$ with~$\alpha^\uparrow[\delta']$ and~$\beta^\uparrow[\delta']$, respectively.
    For each point~$x$ of~$T_1$, let~$v_x \in V(T_1)$ be the lower endpoint of the edge that contains~$x$; if~$x$ is a vertex then~$v_x = x$.
    Moreover, let~$w_x$ be the lower endpoint of the edge that contains~$\alpha(v_x)$.
    The value~$f_2(w_x) - f_1(v_x)$ is a critical value, so by choice of~$\delta'$ we know that~$f_2(w_x) - f_1(v_x) \le \delta'$.
    In other words, $f_2(w_x) \le f_1(x) + \delta'$.
    Therefore, we can safely define~$\alpha'(x) = \an{w_x}{f_1(x) + \delta'}$
    We define the map~$\beta'$ from~$T_2$ to~$T_1$ symmetrically.
    See Figure~\ref{fig:construction-optimal}.

\begin{figure}
    \centering
    \includegraphics{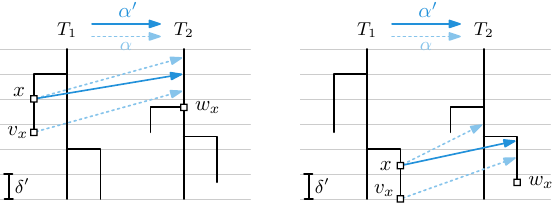}
    \caption{Illustration of the map~$\alpha'$. The image of each point~$x$ is ``pushed down'' to the correct height. On the left there is only one option, but on the right we need to choose the correct subtree.}
    \label{fig:construction-optimal}
\end{figure}

    We first argue that~$\alpha'$ is an up-map; the proof for~$\beta'$ is analogous.
    Let~$x_1$ and~$x_2$ be points of~$T_1$ such that~$x_1 \preceq x_2$.
	If~$v_{x_1} = v_{x_2}$ then we directly obtain~$\alpha'(x_1) \preceq \alpha'(x_2)$.
	Otherwise, we must have~${v_{x_1} \preceq x_1 \preceq v_{x_2}}$.
	Since~$\alpha$ preserves ancestors, it follows that~$\alpha(v_{x_1}) \preceq \alpha(v_{x_2})$ and thus~${w_{x_1} \preceq w_{x_2}}$.
	By construction, we obtain~$\alpha'(x_1) \preceq \alpha'(x_2)$.
	So indeed,~$\alpha'$ preserves ancestors.
	We additionally claim that~$\alpha'(x) \preceq \alpha(x)$ for all~$x$.
	We know that~$\alpha(x) \succeq \alpha(v_x) \succeq w_x$, so the point~$w_x$ is a descendant of both~$\alpha(x)$ and~$\alpha'(x)$.
    By construction, it holds that~$f_2(\alpha'(x)) \le \delta' \le f_2(\alpha(x))$.
    Our claim follows.
	
    It remains to show that~$(\alpha', \beta')$ is an interleaving.
	We first show that~$\beta'(\alpha'(v)) \succeq v$ for all vertices~$v$ of~$T_1$.
	Let~$v'$ be the lower endpoint of the edge that contains~$\beta(\alpha(v))$.
	Then we know that the value~$\delta^* \coloneqq \tfrac{1}{2}(f_1(v') - f_1(v))$ is a critical value that satisfies~$\delta^* \le \delta$.
	By choice of~$\delta'$ it then follows that~$\delta^* \le \delta'$.
    Let~$x \in T_1$
    For all points~$y$ of~$T_2$, we have~$\beta'(y) \preceq \beta(y)$.
    Taking~$y = \alpha'(x)$, we get~$\beta'(\alpha'(x)) \preceq \beta(\alpha'(x))$.
    Similarly, we have~$\alpha'(x) \preceq \alpha(x)$.
    Since~$\beta$ preserves ancestors, it follows that~$\beta(\alpha'(x)) \preceq \beta(\alpha(x))$.
    Combining, we get~$\beta'(\alpha'(x)) \preceq \beta(\alpha(x))$ for all~$x \in T_1$.
	In particular, for~$x = v$, we obtain~$\beta'(\alpha'(v)) \preceq \beta(\alpha(v))$.
	Moreover, by construction we know that the height of~$\beta'(\alpha'(v))$ is exactly~$f_1(v) + 2\delta' \ge f_1(v) + 2\delta^*$.
	It follows that~$\beta'(\alpha'(v))$ is an ancestor of~$v'$ and thus also of~$v$.
	Now, for an arbitrary point~$x$ of~$T_1$, consider the lower endpoint~$v_x$ of the edge that contains~$x$.
	Both~$\alpha'$ and~$\beta'$ preserve ancestors, so we get~$\beta'(\alpha'(x)) \succeq \beta'(\alpha'(v_x)) \succeq v_x$. 
	It follows that~$\beta'(\alpha'(x)) \succeq x$.
	An analogous argument shows that $\alpha'(\beta'(y)) \succeq y$ for all~$y \in T_2$.
\end{proof}

\noindent
We now argue that an optimal interleaving always exists.
Let~$\delta = \intdist{}(T_1, T_2)$, and let~$\delta_1, \delta_2 \in \Delta(T_1, T_2)$ be the greatest value that satisfies~$\delta_1 \le \delta$ and the smallest value that satisfies~$\delta_2 > \delta$, respectively.
By Lemma~\ref{lem:interleaving-equals-compatible} we know that for all~$\eps > 0$ there is an interleaving with shift~$\delta + \eps$.
In particular, for~$\eps$ small enough, there is an interleaving~$\I$ with shift~$\delta + \eps \in [\delta_1, \delta_2]$.
If~$\I$ is optimal, we are done.
Otherwise, by Lemma~\ref{lem:existence-optimal-interleaving} it follows that there exists an interleaving with shift at most~$\delta_1$, which means that~$\delta_1 = \delta$.
Theorem~\ref{thm:existence-optimal-interleaving} directly follows.

\begin{restatable}{theorem}{existenceoptimalinterleaving}
\label{thm:existence-optimal-interleaving}
	Any two merge trees~$T_1$ and~$T_2$ admit an optimal interleaving.
\end{restatable}

\subsection{Residual Interleaving Distance}\label{sec:residual-interleaving-distance}
Fix a partial interleaving~$\P = (\phi, \psi)$.
We refer to a (partial) interleaving that extends~$\P$ as a (\emph{partial})~$\P$-\emph{extension}.
We first observe that there always exists a complete~$\P$-extension.
Indeed, the maps that take every point of~$T_1$ to the root at infinity of~$T_2$, and every point of~$T_2$ to the root at~$\infty$ of~$T_1$ form a complete interleaving that extends~$\P$.

\begin{observation}\label{obs:existence-residual-matching}
	Any partial interleaving $\P$ admits a complete~$\P$-extension.
\end{observation}

\noindent
Our goal is to find a ``tightest possible'' complete~$\P$-extension, that is, an interleaving that extends~$\P$ and minimizes the shift of the remaining arrows.
For some intuition, consider the example in~Figure~\ref{fig:residual-shift}.
The map~$\phi$ consists of a single arrow~$(x, y)$.
Intuitively, the map~$\alpha_2$ describes a ``tighter'' extension of~$\phi$ than the map~$\alpha_1$.
However, if we only disregard the shift of the arrow~$(x, y)$, the shift of the remaining arrows of~$\alpha_1$ is equal to the (supremum) shift of the remaining arrows of~$\alpha_2$.
The arrow~$(x, y)$ implicitly induces a ``fan'' of arrows that all valid up-maps need to extend.
Specifically, any ancestor of~$x$ needs to be mapped to an ancestor of~$y$.
To distinguish between the maps~$\alpha_1$ and~$\alpha_2$, we therefore not only disregard the shift of~$(x, y)$, but of any arrow that is ``contained'' within such a fan.

\begin{figure}
    \centering
    \includegraphics{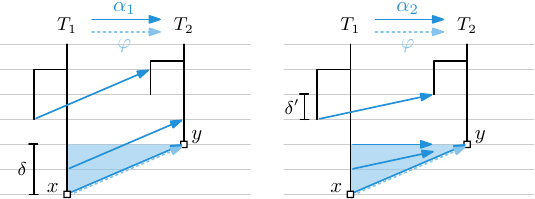}
    \caption{Illustration of the residual shift. The shifts of arrows within the fan~$F[(x, y)]$ (represented by the shaded area) are disregarded, so the~$\phi$-shift of~$\alpha_2$ is less than the~$\phi$-shift of~$\alpha_1$.}
    \label{fig:residual-shift}
\end{figure}

Formally, we define the \emph{fan} of an arrow~$(x, y)$ as the set of arrows~$(x', y)$ for all ancestors~$x'$ of~$x$ with height at most that of~$y$, written,
\[
	F[(x, y)] \coloneqq \{(x',y) \mid x' \in T_1 \text{ such that } x' \succeq x \text{ and } f_1(x') \le f_2(y)\}.
\]
The \emph{fan} of a collection of arrows~$A$, denoted~$F[A]$, is the union of all fans~$F[a]$ for~$a \in A$.

We use fans to define the~$\phi$-\emph{residual shift}, or~$\phi$-\emph{shift} for short.
Recall that we use~$A_\phi$ to denote the set of all arrows~$(x, \phi(x))$.
For an arrow~$a$ from~$T_1$ to~$T_2$, the~$\phi$-residual shift of~$a$, denoted $\shift{\phi}(a)$, is zero if~$a$ lies within~$F[A_\phi]$, and is~$\shift{}(a)$ otherwise.
Similar to the definition of shift in Eq.~\ref{eq:shift}, we define the~$\phi$-residual shift of an up-map~$\phi'$, denoted $\shift{\phi}(\phi')$, as the supremum~$\phi$-residual shift over all arrows of~$\phi'$.
That is,
\[
\shift{\phi}(x, y) \coloneqq \begin{cases}
    0 & \text{if } (x, y) \in F[A_\phi],\\
    \shift{}(x, y) & \text{else,}
\end{cases} \qquad
\shift{\phi}(\phi') \coloneqq \sup\{\shift{\phi}(a) \mid a \in A_{\phi'}\}.
\]

Consider again the example in Figure~\ref{fig:residual-shift}.
We disregard the shifts of all arrows within the fan~$F[(x, y)]$.
As a result, the~$\phi$-shift of~$\alpha_2$ is~$\delta'$, whereas the~$\phi$-shift of~$\alpha_1$ is~$\delta$.

We define the~$\psi$-\emph{(residual) shift} of a partial up-map~$\psi'$ symmetrically.
Finally, for a partial interleaving~$\P' = (\phi', \psi')$, we define the~$\P$-\emph{residual} shift, or~$\P$-\emph{shift}, of~$\P'$ as the maximum of the~$\phi$-shift of~$\phi'$ and the~$\psi$-shift of~$\psi'$, denoted~$\shift{\P}(\P')$.

\begin{definition}\label{def:constrained-interleaving-distance}
	The~$\P$-residual interleaving distance~$\resdist{\P}(T_1, T_2)$ between~$T_1$ and~$T_2$ is the infimum~$\delta$ for which there exists a complete~$\P$-extension whose~$\P$-residual shift is~$\delta$.
\end{definition}

\noindent
When they are clear from context, we omit the arguments~$T_1$ and~$T_2$ and simply write~$\resdist{\P}$ to denote the~$\P$-residual interleaving distance between~$T_1$ and~$T_2$.
We say an extension is \emph{optimal} if its residual shift is equal to the residual interleaving distance.

\subparagraph{Critical values and pairs.}
\begin{figure}
    \centering
    \includegraphics{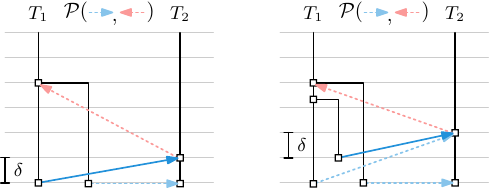}
    \caption{In both examples, the~$\P$-residual interleaving distance is determined by the non-dashed, dark blue arrow. The shift of those arrows is equal to~$\delta$, which in both cases is not a critical value for~$T_1$ and~$T_2$.
    In both cases, however,~$\delta$ is a~$\P$-critical value.
    All critical points are marked.}
    \label{fig:critical-value-set}
\end{figure}
For the remainder of this section, we assume that~$\P$ is \emph{finite}: both~$\dom(\phi)$ and~$\dom(\psi)$ consist of finitely many points.
Recall from Section~\ref{sec:preliminaries} that the non-residual interleaving distance always attains a value in the set of critical values~$\Delta \coloneqq \Delta(T_1, T_2)$.
For the residual interleaving distance, this set~$\Delta$ is not always sufficient (see Figure~\ref{fig:critical-value-set}).
Therefore, we extend~$\Delta$ to a set of~$\P$-critical values.
Specifically, we ``mark'' all vertices of~$T_1$ and~$T_2$, and all points~$x$ of~$T_1$ and~$y$ of~$T_2$ that appear in an arrow~$(x, y)$ or~$(y, x)$ of~$\P$.
We call these the~$\P$-\emph{critical points} of~$T_1$ and~$T_2$, denoted~$C_1[\P]$ and~$C_2[\P]$ respectively (see Figure~\ref{fig:critical-value-set}).
Since~$\P$ is finite, there are only finitely many critical points.
We define a set~$\Delta_1[\P]$ that contains all height differences between critical points:
\begin{align*}
	\Delta_{1}[\P] &\coloneqq \{|f_1(v) - f_2(w)| \mid v \in C_1[\P], w \in C_2[\P]\}.
\end{align*}
We define the set of~$\P$-\emph{critical values} as the set~$\Delta[\P](T_1, T_2) \coloneqq \Delta_1[\P] \cup \Delta_2 \cup \Delta_3$.
When clear from context, we simply write~$\Delta[\P]$.
Each critical value corresponds to a pair of critical points; we refer to such a pair~$b$ as a~$\P$-\emph{critical pair}.
We distinguish two types; see Figure~\ref{fig:critical-pairs}.

\begin{figure}[b]
    \centering
    \includegraphics{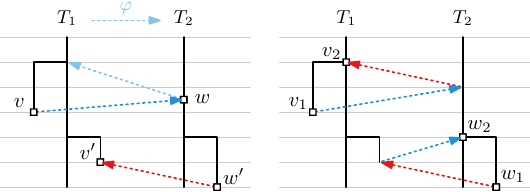}
    \caption{There are two types of critical pairs: arrow critical pairs~$(v, w)$ and~$(w', v')$ (left) and zigzag critical pairs~$(v_1, v_2)$ and~$(w_1, w_2)$ (right).}
    \label{fig:critical-pairs}
\end{figure}

\begin{itemize}
	\item 
		If~$b = (v, w)$ for a point~$v \in C_1[\P]$ and a point~$w \in C_2[\P]$, then~$b$ is an \emph{arrow critical pair} that corresponds to the critical value~$\shift{}(v, w)$.
		We say a partial interleaving~$(\phi, \psi)$ \emph{uses}~$b$ if~$\phi(v) = w$.
		The case that~$b = (w, v)$ is symmetric.
	\item
		If~$b = (v_1, v_2)$ for two vertices~$v_1, v_2$ in~$V(T_1)$, then~$b$ is a \emph{zigzag critical pair} that corresponds to the critical value~$\tfrac{1}{2}(f_1(v_2) - f_1(v_1))$.
		We say a partial interleaving~$(\phi, \psi)$ \emph{uses}~$b$ if there is a point~$y \in T_2$ with~$f_2(y) = f_1(v_1) + \tfrac{1}{2}(f_1(v_2) - f_1(v_1))$ such that~$\phi(v_1) = y$ and~$\psi(y) = v_2$.
		The case that~$b = (w_1, w_2)$ for two vertices~$w_1, w_2$ in~$V(T_2)$ is symmetric.
\end{itemize}

\noindent
Similar to Lemma~\ref{lem:critical-values}, we show that the~$\P$-residual interleaving distance always attains a value in the set of~$\P$-critical values.
We additionally show that there always exists an optimal~$\P$-extension and that any such extension uses a \emph{realizing} critical pair: a critical pair that corresponds to the interleaving distance.
We first prove a generalization of Lemma~\ref{lem:existence-optimal-interleaving}:

\begin{restatable}{lemma}{extensionconstruction}
\label{lem:extension-critical-values}
Let~$\P$ be a finite interleaving.
For any complete~$\P$-extension with~$\P$-shift~$\delta$ that does not use a critical pair that corresponds to~$\delta$, there exists another~$\P$-extension whose~$\P$-shift is equal to a~$\P$-critical value that is strictly less than~$\delta$.
\end{restatable}

\begin{figure}
    \centering
    \includegraphics{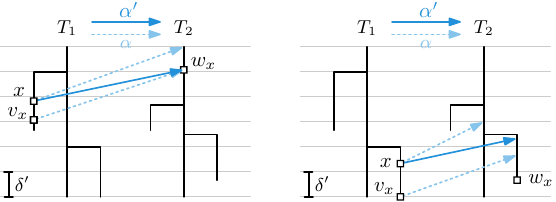}
    \caption{The construction of~$\alpha'$; we need to be careful around the critical points~$v_x$ and~$w_x$.}
    \label{fig:construction-push-down}
\end{figure}

\begin{proof}
Consider a~$\P$-extension~$\I = (\alpha, \beta)$ with~$\P$-shift~$\delta$, and assume~$\P$ does not use a critical pair that corresponds to~$\delta$.
We modify~$\I$ into an extension~$\I' = (\alpha', \beta')$ whose~$\P$-shift is equal to a~$\P$-critical value that is strictly less than~$\delta$.
Let~$\delta' \in \Delta[\P]$ be the greatest value that satisfies~$\delta' < \delta$.
We describe the construction of~$\alpha'$; the construction of~$\beta'$ is symmetric.

We assume, without loss of generality, that each arrow of~$\alpha$ has shift at least~$\delta'$.
Otherwise, we replace~$\alpha$ by~$\alpha^\uparrow[\delta']$.
Since the residual shift of an arrow is bounded by its shift, and since~$\delta' < \delta$, the~$\phi$-residual shift of~$\alpha^\uparrow[\delta']$ is bounded by~$\delta$.
For~$x \in T_1$, let~$v_x \in C_1[\P]$ be the highest descendant critical point of~$x$; if~$x \in C_1[\P]$, then~$v_x = x$.
Moreover, let~$w_x \in C_2[\P]$ be the highest descendant critical point of~$\alpha(v_x)$.
Let~$h_x \coloneqq \max(f_1(x) + \delta', f_2(w_x))$.
We define~$\alpha'(x)$ as the ancestor of~$w_x$ at height~$h_x$, that is,~$\alpha'(x) \coloneqq \an{w_x}{h_x}$.

\subparagraph*{Properties of~$\alpha'$.}
For each point~$x \in T_1$, we have the following three properties:
\begin{enumerate}[(i)]
    \item the height of~$\alpha'(x)$ is at least~$f_1(x) + \delta'$,
    \item $\alpha'(x)$ is an ancestor of~$w_x$,
    \item $\alpha'(x)$ is a descendant of~$\alpha(x)$.
\end{enumerate}
To see that property \enumit{(iii)} holds, observe that as~$x$ is an ancestor of~$v_x$, the point~$\alpha(x)$ must be ancestor of~$\alpha(v_x)$, and thus also of~$w_x$.
If~$h_x = f_2(w_x)$ then~$\alpha'(x) = w_x$ and the property holds.
Otherwise, if~$h_x = f_1(x) + \delta'$, then because the height of~$\alpha(x)$ is at least~$f_1(x) + \delta'$ the property also holds.
We show that~$\alpha'$ is an up-map that extends~$\phi$.

\subparagraph*{$\alpha'$ is an up-map.}
By property~\enumit{(i)} we immediately know that the height of the image of each point~$x$ is at least the height of~$x$.
To show that~$\alpha'$ preserves ancestors, let~$x_1, x_2 \in T_1$ such that~$x_1 \preceq x_2$.
If~$v_{x_1} = v_{x_2}$, then it follows that~$w_{x_1} = w_{x_2}$.
Property~\enumit{(ii)} tells us that both~$\alpha'(x_1)$ and~$\alpha'(x_2)$ are ancestors of~$w_{x_2}$, so as~$f_1(x_1) \le f_1(x_2)$ we infer that~$\alpha'(x_1) \preceq \alpha'(x_2)$.
Otherwise, if~$v_{x_1} \neq v_{x_2}$, then we must have~$v_{x_1} \preceq x_1 \preceq v_{x_2}$.
As~$\alpha$ preserves ancestors, we know that~$\alpha(v_{x_1})$ is a descendant of~$\alpha(v_{x_2})$, and thus also that~$w_{x_1}$ is a descendant of~$w_{x_2}$.
Again, both~$\alpha'(x_1)$ and~$\alpha'(x_2)$ are ancestors of~$w_{x_1}$, so~$\alpha'(x_1) \preceq \alpha'(x_2)$.

Next, to see that~$\alpha'$ extends~$\phi$, let~$x \in \dom(\phi)$.
The point~$x$ is a critical point of~$T_1$, so~$v_x = x$.
Similarly, the point~$\phi(x)$ is a critical point of~$T_2$.
Since~$\alpha$ extends~$\phi$, the point~$\alpha(x)$ is an ancestor of~$\phi(x)$, so~$w_x$ is either equal to or an ancestor of~$\phi(x)$.
Either way, by property~\enumit{(ii)} it follows that~$\alpha'(x)$ is an ancestor of~$\phi(x)$.

Finally, we show that the~$\phi$-shift of~$\alpha'$ is~$\delta'$.
Consider an arrow~$(x, \alpha'(x))$ whose shift is strictly greater than~$\delta'$.
We argue that~$(x, \alpha'(x))$ is in the fan of~$\phi$; this directly implies that the~$\phi$-shift of~$\alpha'$ is~$\delta'$.
By construction, it must be the case that~$h_x = f_2(w_x) > f_1(x) + \delta'$.
In other words, it holds that~$\shift{}(x, w_x) > \delta'$.
The point~$v_x$ is a descendant of~$x$ and thus~$f_1(v_x) \le f_1(x)$.
Therefore, we obtain~$\shift{}(v_x, w_x) \ge \shift{}(x, w_x) > \delta'$.
As both~$v_x$ and~$w_x$ are critical points, the value $\delta^* \coloneqq \shift{}(v_x, w_x)$ is a critical value.
By choice of~$\delta'$ we thus know that~$\delta^* \ge \delta$.
We claim that~$\shift{}(v_x, \alpha(v_x)) > \delta$.
If~$\delta^* = \delta$, then because we assumed that~$\I$ does not use a critical pair that corresponds to~$\delta$, we must have~$\alpha(v_x) \succ w_x$.
The claim follows.
Otherwise, if~$\delta^* > \delta$, then since~$\alpha(v_x) \succeq w_x$ we immediately get~$\shift{}(v_x, \alpha(v_x)) > \delta$.
So our claim holds.

Since the~$\phi$-shift of~$\alpha$ is at most~$\delta$, we know that~$\shift{\P}(v_x, \alpha(v_x)) = 0$.
This means that the arrow~$(v_x, \alpha(v_x))$ is in the fan of~$\phi$.
In other words, there is a descendant~$v'$ of~$v_x$ such that~$v' \in \dom(\phi)$ and~$\phi(v') = w_x$.
As a result, the point~$\alpha(v_x)$ is a critical point of~$T_2$, so~$\alpha(v_x) = w_x$.
Moreover, since~$h_x = f_2(w_x)$, we know that~$\alpha'(x) = w_x$.
Hence, the pair~$(x, \alpha'(x))$ is also in the fan of~$\phi$.
So, $\alpha'$ has~$\phi$-shift~$\delta$.
    
\subparagraph*{$\I'$ is an interleaving.}
Using analogous arguments, the map~$\beta'$ is an up-map that extends~$\psi$ and has~$\psi$-shift~$\delta$.
    It remains to show that~$(\alpha', \beta')$ is an interleaving, that is,~$\beta'(\alpha'(x)) \succeq x$ for all~$x \in T_1$, and~$\alpha'(\beta'(y)) \succeq y$ for all~$y \in T_2$.
    We prove that this statement holds for all~$x \in T_1$; an analogous argument shows the symmetric case.
    We first show that this holds for all vertices of~$T_1$.
    Let~$v \in V(T_1)$.
	From property~\enumit{(i)} we know that the height of~$\beta'(\alpha'(v))$ is at least~$f_1(v) + 2\delta'$.
	Moreover, because~$\I$ is an interleaving, we know that~$\beta(\alpha(v)) \succeq v$ holds.
	Therefore, it suffices to show that~$\beta'(\alpha'(v))$ and~$\beta(\alpha(v))$ lie on the same edge of~$T_1$.
	In the remainder, we use~$y \coloneqq \alpha(v)$ and~$y' \coloneqq \alpha'(v)$.

    \proofsubparagraph{Case $\shift{}(v, y) > \delta$.}
	The~$\phi$-shift of each arrow of~$\alpha$ is at most~$\delta$.
    This means that if~$\shift{}(v, y) > \delta$, then the arrow~$(v, y)$ must lie in the fan of~$\phi$.
	In particular, this means that there is a point~$u \in \dom(\phi)$ such that~$\phi(u) = y$.
    So, the point~$y$ is a critical point of~$T_2$.
	The vertex~$v$ itself is also a critical point, so it follows that~$w_v = y$ and~$h_v = f_2(y)$.
	By construction of~$\alpha'$, we thus have~$y' = y$.
	Similarly, the highest descendant critical point of~$y$ is the point~$y$ itself.
	Let~$v_y \in C_1[\P]$ be the highest descendant critical point of~$\beta(y)$.
    Since all vertices of~$T_1$ are critical points, the point~$\beta(y)$ must lie on the same edge as~$v_y$.
	By construction of~$\beta'$, we know that~$\beta'(y)$ is an ancestor of~$v_y$.
	Combining with property~\enumit{(iii)}, it follows that the points~$\beta'(y') = \beta'(y)$ and~$\beta(y)$ lie on the same edge.

    \proofsubparagraph{Case $\shift{}(v, y) \le \delta$.}
	Recall that~$w_v$ is the highest descendant critical point of~$y$.
	Let~$w_{y'}$ denote the highest descendant critical point of~$y'$.
	We claim that~$w_{y'} = w_v$.
    From property~\enumit{(ii)}, it follows that~$w_{y'}$ is an ancestor of~$w_v$.
	On the other hand, by property~\enumit{(iii)} we know that~$y'$ is a descendant of~$y$.
    As a result, the point~$y'$ either is a descendant of~$w_v$ or it lies on an edge between~$w_v$ and~$y$.
    In both cases, we must have~$w_{y'} \preceq w_v$.
    Combining, we thus obtain~$w_v \preceq w_{y'} \preceq w_v$, which means our claim holds.
    We distinguish two more cases.
    
	\begin{itemize}
		\item 
			First consider the case that~$\shift{}(y, \beta(y)) > \delta$.
            The~$\psi$-shift of~$\beta$ is at most~$\delta$, so the arrow~$(y, \beta(y))$ must lie in the fan of~$\psi$.
            In particular, this means that there is a point~$y^* \in \dom(\psi)$ that is a descendant of~$y$ and satisfies~$\psi(y^*) = \beta(y)$.
            Moreover, we know that the point~$\beta(y)$ is a critical point.
            The point~$y^*$ is a critical point of~$T_2$, so~$y^*$ must be a descendant of~$w_v$.
            Since~$w_{y'} = w_v$, it follows that~$y^*$ is also a descendant of~$w_{y'}$.
            From property~\enumit{(ii)} it then follows that~$\beta'(y')$ is an ancestor of~$\beta(y)$.
            Putting everything together, we obtain~$\beta'(y') \succeq \beta(y) \succeq v$.
		\item
			Lastly, we consider the case that~$\shift{}(y, \beta(y)) \le \delta$.
            Let~$v^* \in V(T_1)$ be the lower endpoint of the edge that contains~$\beta(y)$.
			Both~$v$ and~$v^*$ are vertices, so the value~$\delta^* \coloneqq \tfrac{1}{2}(f_1(v^*) - f_1(v))$ is a critical value.
            Since~$\beta(y) \succeq v^*$, we obtain~$\delta^* \le \tfrac{1}{2}(\shift{}(v, y) + \shift{}(y, \beta(y)))$.
			In our current case, we assumed both summands to be at most~$\delta$, so it follows that~$\delta^* \le \delta$.
            We claim that~$\delta^* < \delta$.
            Otherwise, if~$\delta^* = \delta$, then both~$\shift{}(v, y)$ and~$\shift{}(y, \beta(y))$ must be equal to~$\delta$, and we get~$\beta(y) = v^*$.
            However, we assumed that~$\I$ does not use any critical pairs that correspond to~$\delta$.
            Hence, it holds that~$\delta^* < \delta$.
			By choice of~$\delta'$, we thus have~$\delta^* \le \delta'$.
			By property \enumit{(iii)} for the map~$\beta'$, we get~$\beta'(y') \preceq \beta(y')$.
            Similarly, by property~\enumit{(iii)} for the map~$\alpha'$, we get~$y' \preceq y$.
            Since~$\beta$ preserves ancestors, it follows that~$\beta(y') \preceq \beta(y)$.
            Combining, we obtain~$\beta'(y') \preceq \beta(y') \preceq \beta(y)$.
			Lastly, by property~\enumit{(i)}, we know that the height of~$\beta'(y')$ is at least~$f_1(v) + 2\delta'$.
			So, the point~$\beta'(y')$ lies between the points~$v^*$ and~$\beta(y)$.
            In particular, this means that~$\beta(y)$ and~$\beta'(y')$ lie on the same edge.
            We conclude that~$\beta'(y')$ must be an ancestor of~$v$.
	\end{itemize}

    \noindent
    This shows that~$\beta'(\alpha'(v))$ is an ancestor of~$v$.
    To conclude the proof, consider a point~$x$ of~$T_1$ and let~$v_x \in V(T_1)$ be the lower endpoint of the edge that contains~$x$.
	Since~$\alpha'$ preserves ancestors, we obtain~$\alpha'(x) \succeq \alpha'(v_x)$.
	Similarly, since~$\beta'$ preserves ancestors, we get~$\beta'(\alpha'(x)) \succeq \beta'(\alpha'(v_x))$.
	As~$v_x$ is a vertex, we know that~$\beta'(\alpha'(v_x)) \succeq v_x$.
	Combining, we get~$\beta'(\alpha'(x)) \succeq v_x$.
	Lastly, we know that~$x \succeq v_x$, so it follows that~$\beta'(\alpha'(x)) \succeq x$.
\end{proof}

\noindent
It directly follows that the~$\P$-residual interleaving distance uses a realizing critical pair and that the~$\P$-residual interleaving distance attains a~$\P$-critical value.

\begin{corollary}\label{cor:using-critical-pair}
    If~$\P$ is finite, then any optimal~$\P$-extension uses a realizing critical pair.
\end{corollary}

\begin{corollary}\label{cor:critical-values}
	If~$\P$ is finite, then $\resdist{\P}(T_1, T_2) \in \Delta[\P](T_1, T_2)$.
\end{corollary}

\noindent
Let~$\delta_1 \in \Delta[\P]$ be the greatest value that satisfies~$\delta_1 \le \resdist{\P}$ and let~$ \delta_2 \in \Delta[\P]$ be the smallest value that satisfies~$\delta_2 > \resdist{\P}$.
By definition of~$\resdist{\P}$, for all~$\eps > 0$ there is a~$\P$-extension with~$\P$-shift~$\resdist{\P} + \eps$.
In particular, for~$\eps$ small enough, there is a~$\P$-extension~$\I$ with~$\P$-shift~~$\delta_1 < \shift{\P}(\I) < \delta_2$.
In other words, there are no critical pairs that correspond to~$\shift{\P}(\I)$. 
This means that~$\I$ cannot use a critical pair that corresponds to the~$\P$-shift of~$\I$.
By Lemma~\ref{lem:extension-critical-values}, it then follows that there exists a~$\P$-extension with~$\P$-shift~$\delta_1 = \resdist{\P}$.

\begin{theorem}\label{thm:existence-optimal-extension}
	Any finite partial interleaving~$\P$ between any two merge trees~$T_1$ and~$T_2$ admits an optimal~$\P$-extension.
\end{theorem}

\subparagraph{Isolated extensions.}
We conclude this section by arguing that there always exist ``isolated'' optimal extensions, that is, extensions that satisfy two finiteness properties.
Specifically, this is useful when reasoning about ``greatest'' shifts and ``highest'' ancestor.
Formally, assume the~$\P$-residual interleaving distance is strictly positive, that is,~$\resdist{\P} > 0$.
We say an optimal~$\P$-extension~$\I$ is \emph{isolated} if (1) the~$\P$-shift of only finitely many arrows of~$\I$ is equal to~$\resdist{\P}$, and (2) there exists some $\eps>0$ such that for each arrow $(x,y)$ of $\I$ with $\P$-shift at least $\resdist{\P}-\eps$, there exists a descendant $x'$ of $x$, such that $(x',y)$ is an arrow in $\I$ with the same target and $\P$-shift $\resdist{\P}$.
We show that there always exists an isolated~$\P$-extension.

\begin{restatable}{lemma}{existenceisolated}
\label{lem:existence-isolated-extension}
  If~$\P$ is finite and $\resdist{\P}>0$, there exists an isolated~$\P$-extension.
\end{restatable}
\begin{proof}
By Theorem~\ref{thm:existence-optimal-extension}, there exists an optimal $\P$-extension $\I=(\alpha,\beta)$.
We use $\I$ to construct an isolated $\P$-extension $\I'=(\alpha',\beta')$.
We describe the construction of the map~$\alpha'$ from~$T_1$ to~$T_2$; the construction of the map~$\beta'$ from~$T_2$ to~$T_1$ is symmetric.

Let $X \coloneqq C_1[\P]\cup\beta(C_2[\P])$ and $Y \coloneqq C_2[\P] \cup \alpha(C_1[\P])$.
There are only finitely many critical points of both~$T_1$ and~$T_2$, so the sets~$X$ and~$Y$ are finite.
Since every leaf is critical, every point $x$ of $T_1$ has a descendant in $X$, and by finiteness a highest descendant in $X$, which we denote $v_x$.
Let~$h_x \coloneqq \max(f_1(x), f_2(\alpha(v_x)))$.
We define~$\alpha'(x) = \an{\alpha(v_x)}{h_x}$.

\proofsubparagraph{$\alpha'$ is an up-map with~$\phi$-shift at most~$\resdist{\P}$.}
Let~$x \in T_1$.
We have~$h_x \ge f_1(x)$, so by construction~$f_2(\alpha'(x)) \ge f_1(x)$.
To see that~$\alpha'$ preserves ancestors, let~$x_1$ and~$x_2$ be points of~$T_1$ and suppose~$x_2$ is an ancestor of~$x_1$.
Then~$v_{x_2}$ is an ancestor of~$v_{x_1}$.
By construction, it then follows that~$\alpha'(x_2)$ is an ancestor of~$\alpha(v_{x_2})$, and hence also of~$\alpha(v_{x_1})$.
Moreover, since~$f_1(x_2) \ge f_1(x_1)$, it follows that~$h_{x_2} \ge h_{x_1}$.
We conclude that~$\alpha'(x_2)$ must be an ancestor of~$\alpha'(x_1)$.
So~$\alpha'$ is an up-map.

Next, we argue that~$\alpha'$ extends~$\phi$.
Let~$x \in \dom(\phi)$.
Then~$x$ is a critical point of~$T_1$, so~$x \in X$ and~$v_x = x$.
By construction,~$\alpha'(v_x) = \alpha(v_x)$, and since~$\alpha$ extends~$\phi$ it follows that~$\alpha'(x)$ is ancestor of~$\phi(x)$.
So indeed,~$\alpha'$ extends~$\phi$.
Lastly, to prove that the~$\phi$-shift of~$\alpha'$ is at most~$\resdist{\P}$, it suffices to show that for all points~$x$ it holds that~$\alpha'(x)$ is a descendant of~$\alpha(x)$.
If~$h_x = f_1(x)$ this immediately follows.
Otherwise,~$h_x = f_2(\alpha(v_x))$, and we know that~$\alpha'(x) = \alpha(v_x)$.
The map~$\alpha$ preserves ancestors; since~$v_x$ is a descendant of~$x$, it follows that the point~$\alpha'(x)$ is a descendant of~$\alpha(x)$.
Since $\I$ is an optimal interleaving, we conclude that the~$\P$-shift of~$\alpha'$ is bounded by~$\resdist{\P}$. 

\proofsubparagraph{$\I'$ is an interleaving.}
For any vertex $v$, its highest descendant in the set~$X$, is~$v$ itself.
So, by construction, we have~$\alpha'(v) = \alpha(v)$.
Moreover, as the vertex~$v$ is a critical point,~$\alpha(v)$ is in the set~$Y$.
So, by construction of~$\beta'$, we have~$\beta'(\alpha'(v)) = \beta(\alpha(v))$.
Therefore, because~$\I$ is an interleaving,~$\beta'(\alpha'(v))$ is an ancestor of~$v$.
Now, for an arbitrary point~$x \in T_1$, let~$v \in V(T_1)$ be the lower endpoint of the edge that contains~$x$.
As described above, the point~$\beta'(\alpha'(v))$ is an ancestor of~$v$.
Moreover, since~$\alpha'$ preserves ancestors,~$\alpha'(x)$ is an ancestor of~$\alpha'(v)$.
Similarly, since~$\beta'$ preserves ancestors,~$\beta'(\alpha'(x))$ is an ancestor of~$\beta'(\alpha'(v))$.
Combining, we get that~$\beta'(\alpha'(x))$ is an ancestor of~$v$.
As~$x$ is an ancestor of~$v$, and~$f_1(\beta'(\alpha'(x)) \ge f_1(x)$, we obtain that~$\beta'(\alpha'(x))$ is an ancestor of~$x$.
Therefore~$\I'$ is an interleaving.

\proofsubparagraph{$\I'$ is isolated.}
To see that property (1) holds, we want to bound the number of arrows with~$\P$-shift~$\resdist{\P}$.
Let~$x \in T_1$.
By construction, the point~$\alpha'(x)$ satisfies either (a) $f_2(\alpha'(x)) = f_1(x)$ or (b) $\alpha'(x) = \alpha(v_x)$.
In case~(a), the~$\P$-shift of the arrow~$(x, \alpha'(x))$ is~$0$, so it suffices to bound case (b).
By definition,~$v_x$ is a critical point.
Since~$\P$ is finite, there are only finitely many critical points.
Hence, there are only finitely many points~$y$ of~$T_2$ that can satisfy~$y = \alpha'(x)$.
For any such point~$y$, each edge of~$T_1$ contains at most one point for which the corresponding arrow has shift~$\resdist{\P}$.
Hence, there are only finitely many arrows of~$\alpha'$ with~$\P$-shift~$\resdist{\P}$.
Analogously, there are also only finitely many arrows of~$\beta'$ with~$\P$-shift~$\resdist{\P}$.

Lastly, we show that property (2) holds.
Let~$\eps > 0$ be such that there is no point~$v \in X$ for which~$\resdist{\P} - \eps < f_2(\alpha(v)) - f_1(v) < \resdist{\P}$, and such that~$\resdist{\P} -\eps > 0$.
Since there are only finitely many points in~$X$, such an~$\eps$ exists.
Consider an arrow~$a$ of~$\I'$ with~$\P$-shift at least~$\resdist{\P} - \eps$.
Without loss of generality assume~$a = (x, \alpha'(x))$ for some~$x \in T_1$.
Since the~$\P$-shift of~$a$ is positive, we have~$h_x = f_2(\alpha(v_x))$.
In other words, we have~$\alpha'(x) = \alpha(v_x)$.
As~$v_x$ is a descendant of~$x$, it holds that~$f_2(\alpha'(x)) - f_1(v_x) > f_2(\alpha'(x)) - f_1(x)$.
Moreover, by construction, we have~$\alpha'(v_x) = \alpha(v_x)$, and thus~$\alpha'(v_x) = \alpha'(x)$.
As~$v_x \in X$, by our choice of~$\eps$, the shift of~$(v_x, \alpha'(v_x))$ is equal to~$\resdist{\P}$.
We conclude that property (2) holds.
\end{proof}

\section{Locally Correct Interleavings}\label{sec:locally-correct-interleavings}
The interleaving distance is a bottleneck distance: it yields a single value that quantifies the worst discrepancy between the two merge trees.
While this global value is useful for comparison, we are interested in the actual interleavings that realize the distance.
Importantly, we seek interleavings that are not only globally optimal, but also locally meaningful.
Specifically, if a part of the interleaving is fixed, we want the remainder to be optimal relative to this fixed part.
To formalize this notion of local optimality, we build upon the residual interleaving distance introduced in the previous section.

Let~$(T_1, f_1)$ and~$(T_2, f_2)$ be two merge trees, and fix a partial interleaving~$\P = (\phi, \psi)$.
For a subset~$S_1 \subseteq \dom(\phi)$, we define the \emph{restriction} of~$\phi$ to~$S_1$ as the map~$\phi\vert_{S_1} \colon S_1 \to T_2$ given by~$\phi\vert_{S_1}(x) \coloneqq \phi(x)$ for all points~$x$ of~$S_1$.
Any restriction of~$\phi$ is a partial up-map, and~$\phi$ extends all of its restrictions.
The \emph{restriction} of~$\psi$ to a subset~$S_2 \subseteq T_2$ is defined symmetrically.
Lastly, the pair~$\R = (\phi\vert_{S_1}, \psi\vert_{S_2})$ is a partial interleaving such that~$\P$ extends~$\R$.
We call~$\R$ a \emph{restriction} of~$\P$.
A complete interleaving~$\I$ is \emph{locally correct} if for all restrictions~$\R$ of~$\I$, the~$\R$-residual shift of~$\I$ is equal to the~$\R$-residual interleaving distance.

\begin{definition}\label{def:locally-correct}
	A complete interleaving~$\mathcal{I} = (\alpha, \beta)$ is \emph{locally correct} if and only if for all~$S_1 \subseteq T_1$ and~$S_2 \subseteq T_2$, the restriction~$\mathcal{R} = (\alpha\vert_{S_1}, \beta\vert_{S_2})$ satisfies
	\begin{equation}\label{eq:locally-correct}
	\resdist{\R}(T_1, T_2) = \shift{\mathcal{R}}(\mathcal{I}).
	\end{equation}
\end{definition}

\begin{figure}[b]
    \centering
    \includegraphics{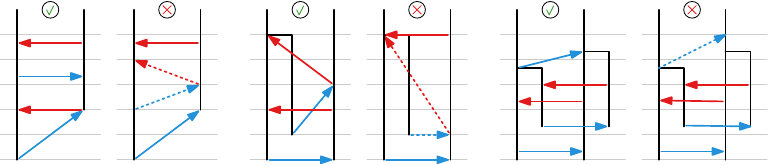}
    \caption{Three instances of~$T_1$ and~$T_2$ that together showcase all the types of critical pairs. For each instance, we show an interleaving that is locally correct (\raisebox{-2pt}{\includegraphics{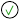}}), and one that is not~(\raisebox{-2pt}{\includegraphics{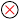}}). The non-dashed arrows in the~\raisebox{-2pt}{\includegraphics{no-icon}} cases form restrictions for which Equation~\eqref{eq:locally-correct} does not hold.
    }
    \label{fig:locally-correct-examples}
\end{figure}

\noindent
See Figure~\ref{fig:locally-correct-examples} for some examples of interleavings that are (not) locally correct.

Let~$\I$ be an arbitrary, not necessarily locally correct, interleaving.
For any restriction~$\R$ of~$\I$, we know that~$\I$ extends~$\R$.
Hence, the~$\R$-residual interleaving distance is at most the~$\R$-shift of~$\I$.
To see that an interleaving~$\I$ is locally correct, it hence suffices to show that for all restrictions~$\R$ of~$\I$ it holds that~$\resdist{\R}(T_1, T_2) \ge \shift{\R}(\I)$.

\begin{observation}\label{obs:locally-correct}
	For any restriction~$\R$ of~$\I$, it holds that~$\resdist{\R}(T_1, T_2) \le \shift{\R}(\I)$.
\end{observation}

\noindent
Suppose~$\I$ is locally correct.
Taking~$\R$ as the empty restriction, we see that~$\I$ satisfies~$\intdist{}(T_1, T_2) = \shift{}(\I)$.
In other words, any locally correct interleaving is an optimal interleaving.
The reverse is not necessarily true; not every optimal interleaving is locally correct (see Figure~\ref{fig:good-vs-bad-optimal-interleavings}).
That raises the question: does a locally correct interleaving always exist?
In the remainder of this section, we answer this question affirmatively.

\begin{theorem}\label{thm:existence-locally-correct}
	Any two merge trees~$T_1$ and~$T_2$ admit a locally correct interleaving.
\end{theorem}

\subparagraph{Overview of the proof.}
To prove Theorem~\ref{thm:existence-locally-correct}, we give an explicit construction of a locally correct interleaving.
Specifically, we incrementally build a partial interleaving~$\P$ by augmenting it with ``bottleneck'' arrows.
We show that with every augmentation, the~$\P$-residual interleaving distance strictly decreases, until, after a finite number of iterations, it becomes zero.
To conclude the proof, we then show that any complete interleaving that extends the resulting partial interleaving~$\P$ is locally correct.

\subsection{Bottlenecks}\label{sec:bottlenecks}
The interleaving distance is determined by the greatest shift among all arrows of an optimal interleaving.
Locally, we might improve such an interleaving by ``pushing down'' some of the arrows, thereby decreasing their shifts. 
However, we cannot do this for all arrows; otherwise, we would obtain an interleaving whose shift is strictly smaller than the interleaving distance.
We are interested in a \emph{bottleneck}: a set of arrows that we cannot ``push down''.
In this section, we give a characterization of such a bottleneck, and we argue that, under some mild assumptions, all arrows in a bottleneck correspond to realizing critical pairs.

We first define the \emph{relative difference} of two, not necessarily related, partial up-maps~$\phi$ and~$\phi'$.
Specifically, consider the set of points for which either~$\phi$ is defined and~$\phi'$ is not, or for which both are defined but where they do not agree:
\[
S \coloneqq \dom(\phi) \setminus \{x \in \dom(\phi) \cap \dom(\phi') \mid \phi(x) = \phi'(x)\}. 
\]
We define the relative difference of~$\phi$ and~$\phi'$ as the restriction of~$\phi$ to~$S$, denoted~${\phi \setminus \phi' \coloneqq \phi \vert_{S}}$.
Similarly, for two partial interleavings~$\P = (\phi, \psi)$ and~$\P' = (\phi', \psi')$, we define the \emph{relative difference}~$\P \setminus \P'$ as the pair of relative differences~$(\phi \setminus \phi', \psi \setminus \psi')$.
The resulting pair is a restriction of~$\P$, so it directly follows that any relative difference is a partial interleaving.

\subparagraph{Augmentations.}
Fix a partial interleaving~$\P = (\phi, \psi)$.
In the remainder of this section, we assume that~$\P$ is finite, and that the~$\P$-residual interleaving distance is strictly positive, that is,~$\resdist{\P}> 0$.
A~$\P$-extension~$\Q$ is a~$\P$-\emph{augmentation} if (i) the~$\P$-residual shift of~$\Q$ is at most~$\resdist{\P}$, that is,~$\shift{\P}(\Q) \le \resdist{\P} $, and (ii) the~$\Q$-residual interleaving distance is strictly less than~$\resdist{\P}$, that is,~$\resdist{\Q} < \resdist{\P}$.
We say that~$\Q$ is \emph{minimal} if it does not extend any other~$\P$-augmentation.

\begin{restatable}{lemma}{existenceminimalaugmentation}
\label{lem:existence-minimal-augmentation}
    Any finite interleaving~$\P$ with~$\resdist{\P} > 0$ admits a finite minimal~$\P$-augmentation.
\end{restatable}
\begin{proof}
    By Lemma~\ref{lem:existence-isolated-extension}, there exists an isolated~$\P$-extension.
    Let~$\I$ be such an extension.
    Then in~$\I$, only finitely many arrows of~$\P$ have~$\P$-shift~$\resdist{\P}$.
    Let $\Q$ be the partial interleaving that consists of the arrows of~$\P$ and the arrows of~$\I$ with~$\P$-shift~$\resdist{\P}$.
    If a point~$x$ of~$T_1$ occurs in two arrows~$(x, y_1)$ and~$(x, y_2)$, with~$f_2(y_1) \le f_2(y_2)$, we take the arrow~$(x, y_2)$.
    We choose arrows symmetrically if a point of~$T_2$ occurs in two arrows.
    Since~$\I$ extends~$\P$, it directly follows that~$\Q$ extends~$\P$ as well.
    
    We now argue that~$\Q$ is an augmentation.
    By construction, each arrow of~$\Q$ has~$\P$-shift at most~$\resdist{\P}$, so property (i) holds.
    To see that property (ii) holds; observe that~$\I$ is a~$\Q$-extension.
    Since~$\I$ is isolated, there is an~$\eps > 0$ such that for each arrow~$(x, y)$ of~$\I$ with~$\P$-shift at least~$\resdist{\P} - \eps$, there exists a descendant~$x'$ of~$x$ such that~$(x', y)$ is an arrow of~$\I$.
    By construction, the arrow~$(x',y)$ is an arrow of~$\Q$.
    In other words, the~$\Q$-shift of the arrow~$(x, y)$ is~$0$.
    Additionally, the~$\Q$-shift of each arrow of~$\P$ is at most~$\resdist{\P}$.
    Together, it follows that all arrows of~$\I$ with shift at least~$\resdist{\P} - \eps$ have~$\Q$-shift~$0$.
    Hence, the~$\Q$-residual interleaving distance is at most~$\resdist{\P}-\eps < \resdist{\P}$.
    We conclude that~$\Q$ is an augmentation with a finite number of arrows, which implies that there exists a finite minimal augmentation.
\end{proof}

\begin{figure}
    \centering
    \includegraphics{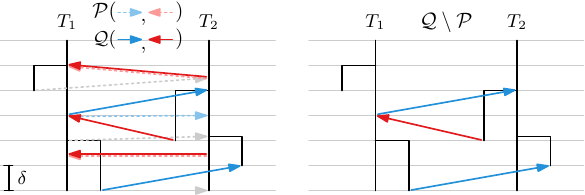}
    \caption{A partial interleaving~$\P$ with~$\resdist{\P} = \delta$, and a minimal augmentation~$\Q$ with~$\resdist{\Q} < \delta$; the dashed gray arrows showcase an optimal~$\Q$-extension (left). The bottleneck~$\Q \setminus \P$ (right).}
    \label{fig:bottleneck}
\end{figure}

\noindent
Let~$\Q$ be a minimal~$\P$-augmentation.
We refer to the relative difference~$\Q \setminus \P$ as a~$\P$-\emph{bottleneck}.
Intuitively, a bottleneck captures a minimal set of arrows that we need to ``add'' to decrease the residual interleaving distance.
See Figure~\ref{fig:bottleneck} for an illustration.

Recall (the proof of) Lemma~\ref{lem:extension-critical-values}.
If we apply the lemma to a given~$\P$-extension~$\I$ that does not use a realizing critical pair, we can construct another~$\P$-extension~$\I'$ that has strictly smaller shift.
The exact same proof still works for a slightly stronger statement: it suffices to assume that the relative difference~$\I \setminus \P$ does not use a realizing critical pair.
\begin{observation}\label{obs:realizing-critical-pair}
Let~$\P$ be a finite interleaving and let~$\I$ be a~$\P$-extension.
If~$\I \setminus \P$ does not use a realizing critical pair, there exists a~$\P$-extension~$\I'$ with~$\shift{\P}(\I') < \shift{\P}(\I)$.
\end{observation}
We use this observation to prove the following lemma.

\begin{lemma}\label{lem:bottleneck-critical-pair}
    If~$\P$ is finite and~$\resdist{\P} > 0$, any~$\P$-bottleneck~$\Q \setminus \P$ uses a realizing critical pair.
\end{lemma}
\begin{proof}
    Let~$\I$ be an optimal~$\Q$-extension.
    Since~$\Q$ extends~$\P$, so does~$\I$.
    Moreover, as~$\I$ is optimal, the~$\Q$-shift of~$\I$ is equal to the~$\Q$-residual interleaving distance, that is,~$\shift{\Q}(\I) = \resdist{\Q}$.
    Since~$\Q$ is an augmentation, the shift of all arrows of~$\Q \setminus \P$ is at most~$\resdist{\P}$, and~$\resdist{\Q} < \resdist{\P}$.
    Combining, we obtain~$\shift{\P}(\I) \le \resdist{\P}$.
    In other words,~$\I$ is an optimal~$\P$-extension.
    By Observation~\ref{obs:realizing-critical-pair}, it follows that~$\I \setminus \P$ uses a realizing critical pair~$b$.
    As~$\resdist{\Q} < \resdist{\P}$, we get that the~$\Q$-shift of the arrows that correspond to~$b$ must be~$0$.
    In other words,~$\Q \setminus \P$ uses~$b$.
\end{proof}

\noindent
A direct consequence of Lemma~\ref{lem:bottleneck-critical-pair} is that if, for a given~$\P$-extension~$\Q$, the relative difference~$\Q \setminus \P$ does not use a realizing critical pair, then~$\Q$ cannot be a~$\P$-augmentation.
In other words, if the~$\P$-shift of each arrow of~$\Q$ is strictly less than the~$\P$-residual interleaving distance, then the~$\Q$-residual interleaving distance cannot be less than~$\resdist{\P}$.

\begin{corollary}\label{cor:helper-lemma}
    Let~$\P$ be a finite interleaving with~$\resdist{\P} > 0$, and let~$\Q$ be a~$\P$-extension.
    If~$\shift{\P}(a) < \resdist{\P}(T_1, T_2)$ for all arrows~$a \in A_\Q$, then $\resdist{\Q}(T_1, T_2) \ge \resdist{\P}(T_1, T_2)$.
\end{corollary}

\subparagraph{Dominant matchings.}
We say~$\P$ is \emph{dominant} if the shift of each arrow of~$\P$ is strictly greater than~$\resdist{\P}$.
If~$\P$ is dominant, we can characterize any minimal augmentation~$\Q$.
Recall that we say that~$\Q$ uses~$\P$ if they completely agree on the domains of~$\P.\phi$ and~$\P.\psi$.

\begin{lemma}\label{lem:optimal-uses}
Let~$\P$ be a dominant, finite interleaving with~$\resdist{\P} > 0$.
For any minimal~$\P$-augmentation~$\Q$, (i) $\Q$ uses~$\P$, and (ii) for each arrow~$a$ of~$\Q \setminus \P$, it holds that~$\shift{}(a) = \resdist{\P}$.
\end{lemma}
\begin{proof}
We first show (i).
Suppose for a contradiction that an arrow~$a$ of~$\P$ is not used by~$\Q$.
Because~$\P$ is dominant, the shift of~$a$ is at least~$\resdist{\P}$.
Moreover, since~$\Q$ extends~$\P$, there is an arrow~$a'$ that extends~$a$.
Therefore, the arrow~$a'$ is not in the fan of~$\P$, so the~$\P$-shift of~$a'$ is equal to the shift of~$a'$.
We obtain~$\shift{\P}(\Q) \ge \shift{\P}(a') > \shift{}(a) \ge \resdist{\P}$.
However, $\Q$ is an augmentation: we have reached a contradiction.
Property (i) follows.

Next, we show (ii).
By definition, the shift of~$a$ is at most~$\resdist{\P}$, that is,~$\shift{}(a) \le \resdist{\P}$.
For a contradiction, assume that the shift of~$a$ is strictly less than~$\resdist{\P}$, that is,~$\shift{}(a) < \resdist{\P}$.
In particular, consider the partial interleaving~$\Q'$ obtained after removing~$a$ from~$\Q$.
Since the shift of~$a$ is less than~$\resdist{\P}$, it is not an arrow of~$\P$.
It follows that~$\Q'$ extends~$\P$.
Lastly, any (complete)~$\Q$-extension~$\I$ is trivially also a~$\Q'$-extension.
Since~$\shift{}(a) < \resdist{\P}$, we then know that~$\shift{\Q'}(\I) = \max(\shift{}(a), \shift{\Q}(\I)) \le \resdist{\P}$.
But then~$\Q$ is not a minimal augmentation.
\end{proof}

\begin{restatable}{lemma}{onlycriticalpairs}
\label{lem:only-critical-pairs}
Let~$\P$ be a dominant, finite interleaving with~$\resdist{\P} > 0$.
For any finite minimal~$\P$-augmentation~$\Q$, $\Q \setminus \P$ uses only realizing critical pairs.
\end{restatable}
\begin{proof}
    For a contradiction, assume that~$\Q \setminus \P$ uses an arrow~$a = (x, y)$ that does not correspond to a critical pair.
    By Lemma~\ref{lem:optimal-uses}, we know that the shift of~$a$ is exactly~$\resdist{\P}$.
    Let~$\Q'$ be the partial interleaving obtained by removing~$a$ from~$\Q$.
    By construction,~$\Q$ extends~$\Q'$.
    We argue that~$\Q'$ is an augmentation, thereby contradicting that~$\Q$ is minimal.
    Without loss of generality, assume~$x \in T_1$ and~$y \in T_2$.    
    Let~$\I = (\alpha, \beta)$ be an optimal~$\Q$-extension.
    As~$\Q$ is a~$\P$-augmentation, it follows that~$\I$ is also an optimal~$\P$-extension.

    If~$\I$ uses an arrow~$(x', y)$ for some descendant~$x'\preceq x$, then~$\shift{}(x', y) > \shift{}(x, y) = \resdist{\P}$.
    Because~$\I$ is an optimal~$\Q$-extension, the~$\Q$-shift of~$\I$ is strictly less than~$\resdist{\P}$, so the~$\Q$-shift of~$(x', y)$ must be~$0$.
    In other words, the arrow~$(x', y)$ lies in the fan of~$\Q.\phi$, that is, there exists a descendant~$x''$ of~$x'$ such that the arrow~$(x'', y)$ is used by~$\Q$.
    By construction, the arrow~$(x'', y)$ is also used by~$\Q'$.
    But then, the point~$x''$ is a descendant of~$x$, so the~$\Q'$-shift of~$(x, y)$ is also equal to~$0$.
    Hence,~$\I$ is a~$\Q'$-extension with~$\Q'$-shift strictly smaller than~$\resdist{\P}$.
    This implies that~$\Q'$ is a~$\P$-augmentation, and we reach a contradiction.
    For the remainder, we assume that for all strict descendants~$x'$ of~$x$, the point~$\alpha(x')$ is a strict descendant of~$y$.

    \begin{claim}\label{claim:A}
        For some~$\eps > 0$, each strict descendant~$x'$ of~$x$ satisfies~$f_2(\alpha(x')) < f_2(y) - \eps$.
    \end{claim}
    \begin{claimproof}
    As~$\Q$ is a~$\P$-augmentation, it holds that~$\resdist{\Q} < \resdist{\P}$.
    Moreover, as~$\Q$ is finite, there are only finitely many~$\Q$-critical values.
    By Corollary~\ref{cor:critical-values}, both~$\resdist{\P}$ and~$\resdist{\Q}$ are equal to~$\Q$-critical values.
    Hence, there exists an~$\eps' > 0$ such that~$\resdist{\Q} + \eps' < \resdist{\P}$.
    If there is no strict descendant~$x'$ of~$x$ with~$f_2(\alpha(x')) \ge f_2(y) - \eps'$, then our claim immediately holds if we take~$\eps = \eps'$.
    Otherwise, let~$y' = \alpha(x')$; we assumed that~$y' \neq y$.
    It follows that the arrow~$(x', y')$ satisfies~$\shift{}(x', y') > \resdist{\Q}$.
    Since~$\I$ is an optimal~$\Q$-extension, we must have~$\shift{\Q}(x', y') \le \resdist{\Q}$; we obtain~$\shift{\Q}(x', y') = 0$.
    So, the arrow~$(x', y')$ lies in the fan of~$\Q.\phi$.
    In other words, there is a descendant~$x''$ of~$x'$ such that the arrow~$(x'', y')$ is used by~$\Q$.
    Now, the points~$y'$ and~$y$ are two distinct~$\Q$-critical points.
    There are only finitely many~$\Q$-critical points, so there exists~$\eps''$ such that~$f_2(y') + \eps'' < f_2(y)$.
    \end{claimproof}

    \noindent
    A direct consequence of Claim~\ref{claim:A} is that if~$x$ is not a leaf, the point~$y$ must have a descendant.
    If~$x$ is a leaf, then~$y$ cannot be a leaf, since otherwise the pair~$(x, y)$ is an arrow critical pair.
    Either way, $y$ cannot be a leaf.
    We construct a pair~$\I' = (\alpha', \beta')$ that does not use~$a$.
    For~$\eps > 0$, we use~$X^\eps$ to denote the set of ancestors~$x'$ of~$x$ with~$f_1(x') - f_1(x) < \eps$, and we use~$Y_\eps$ to denote the set of descendants~$y'$ of~$y$ with~$f_2(y) - f_2(y') < \eps$.
    
    \proofsubparagraph{Case 1: $y$ is a critical point.}
    If~$y$ is a vertex, then for any critical point~$v$ of~$T_1$, the pair~$(v, y)$ corresponds to a critical pair.
    In particular, this means that~$x$ cannot be a critical point.
    By Claim~\ref{claim:A}, we know that there exists an~$\eps_1 > 0$ such that (i) all descendants~$x'$ of~$x$ satisfy~$f_2(\alpha(x')) < f_2(y) - \eps_1$.
    Since~$x$ is not a critical point, it cannot be a leaf.
    Hence, the point~$x$ has a strict descendant~$x'$, which means that there must be a strict descendant~$y^*$ of~$y$ at height~$f_2(y^*) = f_1(y) - \eps_1$ that satisfies~$y^* \succeq \alpha(x')$ for all~$x'$.
    Moreover, let~$\eps_2 > 0$ be such that (ii) there is no~$\Q$-critical point in~$X^{\eps_2}$.
    Lastly, let~$\eps_3 > 0$ such that (iii) $\resdist{\Q} + 2\eps_3 < \resdist{\P}$.
    Set~$\eps \coloneqq \min(\eps_1, \eps_2, \eps_3)$.
    For~$x' \in T_1$ and~$y' \in T_2$, we define the maps~$\alpha'$ and~$\beta'$ as follows:
    \[
        \alpha'(x') \coloneqq \begin{cases}
            \an{y^*}{\max(f_1(x'), f_2(y)-\eps)}, & \text{if } x' \in X^\eps, \\
            \alpha(x') & \text{otherwise,}
        \end{cases} \qquad
        \beta'(y') \coloneqq \begin{cases}
            \an{x}{f_1(x) + \eps}, & \text{if } \beta(y') \in X^\eps, \\
            \beta(y') & \text{otherwise.}
        \end{cases}
    \]
    Since~$f_2(y^*) = f_2(y) - \eps_1 \le f_2(y) - \eps$ and~$f_1(x) \le f_1(x) + \eps$, both~$\alpha'$ and~$\beta'$ are well-defined.

    
    We first argue that~$\alpha'$ is an up-map.
    By construction, we have~$f_2(\alpha(x')) \ge f_1(x')$ for all~$x'$ of~$T_1$.
    To see that~$\alpha'$ preserves ancestors, let~$x_1 \preceq x_2$ be two points of~$T_1$.
    If both~$x_1$ and~$x_2$ lie within, or if both lie outside of~$X^\eps$, then we immediately obtain~$\alpha'(x_1) \preceq \alpha'(x_2)$.
    If~$x_1$ lies within and~$x_2$ lies outside of~$X^\eps$, we get~$\alpha'(x_1) \preceq \alpha(x_1)$ and~$\alpha'(x_2) = \alpha(x_2)$.
    Since~$\alpha$ preserves ancestors, we obtain~$\alpha'(x_1) \preceq \alpha'(x_2)$.
    Otherwise, if~$x_1$ lies outside of~$X^\eps$ and~$x_2$ lies within~$X^\eps$, then, by choice of~$y^*$, we know that~$\alpha'(x_1)$ is a (strict) descendant of~$y^*$.
    It directly follows that~$\alpha'(x_1) = \alpha(x_1) \preceq \alpha'(x_2)$.
    By property (ii), there are no~$\Q$-critical points, and hence no~$\P$-critical points within~$X^\eps$.
    So,~$\alpha'$ trivially extends~$\P.\phi$.
    By construction, the~$\Q'.\phi$-shift of~$\alpha'$ is strictly less than~$\resdist{\P}$.

    Next, we argue that~$\beta'$ is an up-map.
    Let~$y' \in T_2$.
    Observe that, by construction, we have~$\beta'(y')$ is an ancestor of~$\beta(y')$.
    This directly implies that~$\beta'$ extends~$\P.\psi$ and that~$f_1(\beta'(y')) \ge f_2(y')$.
    To see that~$\beta'$ preserves ancestors, consider two points~$y_1 \preceq y_2$ of~$T_2$.
    The only non-trivial case is if~$\beta(y_1) \in X^\eps$ and~$\beta(y_2) \notin X^\eps$.
    Then, the point~$\beta'(y_2) = \beta(y_2)$ must be ancestor of every point of~$X^\eps$.
    In particular, this means that~$\beta'(y_2) \succeq \an{x}{f_1(x)+\eps} = \beta'(y_1)$.
    So~$\beta'$ is an up-map.
    Lastly, we show that the~$\P.\psi$-shift of~$\beta'$ is bounded by~$\resdist{\P}$.
    Let~$y' \in T_2$.
    If~$\beta(y') \notin X$, then we directly obtain~$\shift{\P.\psi}(y, \beta'(y)) = \shift{\P.\psi}(y, \beta(y)) \le \resdist{\P}$. 
    Otherwise, if~$\beta(y') \in X^\eps$, then we get~$f_1(\beta'(y')) = f_1(x) + \eps$.
    There is no critical point in~$X^\eps$, so we must have~$\shift{}(y', \beta(y')) < \resdist{\Q}$.
    So, by property (iii), it follows that~$\shift{}(y', \beta'(y')) \le \shift{}(y', \beta(y')) + \eps < \resdist{\Q} + \eps < \resdist{\P}$.
    As a result, the~$\Q'.\psi$-shift of~$\beta'$ is strictly less than~$\resdist{\P}$.

    Finally, we show that~$\I'$ is an interleaving.
    Let~$x' \in T_1$.
    If~$x' \notin X^\eps$, we get~$\beta'(\alpha'(x')) = \beta'(\alpha(x')) \succeq \beta(\alpha(x')) \succeq x'$.
    Otherwise, if~$x' \in X^\eps$, then let~$v$ be the lower endpoint of the edge that contains~$x'$.
    We know that~$v$ is not within~$X^\eps$, so~$\beta'(\alpha'(v))$ is an ancestor of~$v$.
    As~$\alpha'$ is an up-map, we get~$\alpha'(x')$ is an ancestor of~$\alpha'(v)$.
    Similarly, since~$\beta'$ is an up-map, we get~$\beta'(\alpha(x'))$ is an ancestor of~$\beta'(\alpha'(v))$.
    So,~$\beta'(\alpha'(x'))$ is an ancestor of~$v$; it follows that it is also an ancestor of~$x'$.
    Lastly, let~$y' \in T_2$.
    If~$\beta(y') \notin X^\eps$, we get~$\alpha'(\beta'(y')) = \alpha(\beta(y')) \succeq y'$.
    Otherwise, if~$\beta(y') \in X^\eps$, we get~$\beta'(y') = \an{x}{f_1(x) + \eps} \succeq \beta(y')$.
    The point~$\beta(y')$ is not within~$X^\eps$, so as~$\alpha'$ preserves ancestors,~$\alpha'(\beta(y'))$ is an ancestor of~$y'$.
    Therefore~$\alpha'(\beta'(y'))$ is an ancestor of~$y'$.
    So,~$\I'$ is an interleaving.
    Moreover, the~$\Q'$-shift of~$\I'$ is strictly less than~$\resdist{\P}$.

    \proofsubparagraph{Case 2: $y$ is not a critical point.}
    We choose~$\eps$ as follows.
    First, let~$y^*$ be the highest critical descendant of~$y$; this is a strict descendant of~$y$.
    Let~$\eps_1 > 0$ be such that (i) $f_2(y^*) \le f_2(y) - \eps_1$.
    Secondly, let~$\eps_2 > 0$ be such that (ii) there are no~$\Q$-critical points (beside potentially the point $x$ itself) in the set~$X^{\eps_2}$.
    Such a value exists, because there are only finitely many~$\Q$-critical points.
    Thirdly, let~$\eps_3 > 0$ be such that for all strict descendants~$x'$ of~$x$, it holds that~$f_2(\alpha(x')) < f_2(y) - \eps_3$.
    Such a value~$\eps_3$ exists by Claim~\ref{claim:A}.
    Note that~$x$ does not necessarily have strict descendants; in that case any~$\eps_3 > 0$ suffices.
    Lastly, let~$\eps_4 > 0$ be such that~$\resdist{\Q} + \eps < \resdist{\P}$.
    Let~$\eps \coloneqq \min(\eps_1, \eps_2, \eps_3, \eps_4)$.
    For~$x' \in T_1$ we define the map~$\alpha'$ as follows:
    \[
        \alpha'(x') \coloneqq \begin{cases}
            \an{y^*}{\max(f_1(x'), f_2(y)-\eps)}, & \text{if } x' \in X^\eps, \\
            \alpha(x') & \text{otherwise.}
        \end{cases}
    \]
    Observe that~$f_2(y^*) \le f_2(y) - \eps_1 \le f_2(y) - \eps$.
    Hence, the map~$\alpha'$ is well-defined.
    We now show that~$\alpha'$ is an up-map that extends~$\P.\phi$.

    Let~$x'$ be a point of~$T_1$.
    By construction, we directly have~$f_1(\alpha'(x')) \ge f_1(x')$.
    To see that~$\alpha'$ preserves ancestors, observe that~$\alpha'(x')$ always lies on the same edge as~$\alpha(x')$.
    Moreover, by choice of~$\eps_3$, the point~$y^*$ is an ancestor of all descendants of~$x$.
    It follows that~$\alpha'$ preserves ancestors.
    Next, to see that~$\alpha'$ extends~$\P.\phi$, suppose~$x' \in \dom(\P.\phi)$.
    If~$x' \notin X^\eps$, then we have~$\alpha'(x') = \alpha(x') \succeq \P.\phi(x')$.
    Otherwise, by choice of~$\eps$, we must have~$x' = x$.
    Since~$y^*$ is the highest descendant critical point of~$y$, the point~$\P.\phi(x)$ must be a descendant of~$y^*$.
    Hence, by construction, we indeed get~$\alpha'(x) \succeq y^* \succeq \P.\phi(x)$.
    Lastly, by construction, the~$\Q'.\phi$-shift of~$\alpha'$ is strictly less than~$\resdist{\P}$. 

    To define the map~$\beta'$, we consider two more cases.
    
    \proofsubparagraph{Case 2a: $x$ is not a leaf.}
    If $x$ is not a leaf, we let~$\beta' = \beta$.
    Trivially,~$\beta'$ is an up-map that extends~$\P.\psi$ and whose~$\Q'.\psi$-shift is strictly less than~$\resdist{\P}$.
    It remains to show that the pair~$(\alpha', \beta')$ is an interleaving.

    Let~$x' \in T_1$.
    If~$x' \notin X^\eps$, we have~$\beta'(\alpha'(x)) = \beta(\alpha(x))$, which directly implies that~$\beta'(\alpha'(x))$ is an ancestor of~$x$.
    Otherwise, we have~$x' \in X^\eps$.
    Since~$x$ is not a leaf, there must be a descendant~$x''$ of~$x$.
    In particular, by choice of~$\eps$, it holds that~$\alpha'(x'') = \alpha(x'')$ is a strict descendant of~$\alpha'(x')$.
    The point~$x''$ does not lie within~$X^\eps$, so we have~$\beta'(\alpha'(x'')$ is an ancestor of~$x''$.
    As~$\beta'$ preserves ancestors and~$\alpha'(x')$ is an ancestor of~$\alpha'(x'')$, we obtain that~$\beta'(\alpha'(x'))$ is an ancestor of~$x''$.
    It follows that~$\beta'(\alpha'(x'))$ is an ancestor of~$x'$.
    Lastly, let~$y' \in T_2$; we have~$\beta'(y') = \beta(y')$.
    Moreover, by construction, the points~$\alpha'(\beta(y'))$ and~$\alpha(\beta(y'))$ lie on the same edge.
    Since~$\I$ is an interleaving, it directly follows that~$\alpha'(\beta'(y'))$ is an ancestor of~$y'$.
    
    \proofsubparagraph{Case 2b: $x$ is a leaf.}
    If~$x$ is a leaf, we also modify~$\beta$.
    Specifically, we use~$Y_\eps$ to denote the set of descendants~$y'$ of~$y$ with height~$f_2(y) - f_2(y') < \eps$. 
    For~$y' \in T_2$, we define~$\beta'$ as follows:
    \[
        \beta'(y') \coloneqq \begin{cases}
            \beta(y), & \text{if } y' \in Y_\eps, \\
            \beta(y') & \text{otherwise.}
        \end{cases}
    \]

    We argue that~$\beta'$ is an up-map that extends~$\P.\psi$.
    Let~$y' \in T_2$.
    By construction, we get~$f_1(\beta'(y')) \ge f_1(\beta(y')) \ge f_2(y')$.
    Moreover, if~$y' \in Y_\eps$, its image~$\beta'(y')$ is equal to the image of its lowest ancestor not in~$Y_\eps$; it directly follows that~$\beta'$ preserves ancestors.
    Similarly, since~$\beta'(y')$ is an ancestor of~$\beta(y')$, it follows that if~$y' \in \dom(\P.\psi)$ we have~$\beta'(y')$ is an ancestor of~$\P.\psi(y')$.
    So,~$\beta'$ is an up-map that extends~$\P.\psi$.
    
    Next, we show that the~$\Q'.\psi$-shift of~$\beta'$ is strictly less than by~$\resdist{\P}$.
    Since~$\I$ is an optimal~$\Q$-extension, the~$\Q$-shift of each arrow of~$\I$ is bounded by~$\resdist{\Q}$.
    Let~$y' \in T_2$.
    By construction,~$\Q'.\psi$ is equal to~$\Q.\psi$.
    If~$y'$ does not lie within~$Y_\eps$, the~$\Q'.\psi$-shift of~$(y', \beta'(y'))$ is bounded by~$\resdist{\Q}$, which is strictly less than~$\resdist{\P}$.
    Otherwise, we have~$y' \in Y_\eps$.
    The~$\Q$-shift of the arrow~$(y, \beta(y))$ is bounded by~$\resdist{\Q}$.
    This means that either (a)~$\shift{}(y, \beta(y)) \le \resdist{\Q}$, or (b) the arrow~$(y, \beta(y))$ lies inside the fan of~$\Q.\psi$.
    In case (a), by choice of~$\eps_4$, it follows that~$\shift{}(y', \beta(y)) \le \resdist{\Q} + \eps < \resdist{\P}$.
    In case (b), there is a descendant~$y''$ of~$y$ such that~$y'' \in \dom(\Q.\psi)$ and~$\Q.\psi(y'') = \beta(y)$.
    The point~$y''$ is a~$\Q$-critical point, so it must be a descendant of~$y^*$.
    By choice of~$\eps_1$, it follows that~$y''$ is a descendant of~$y'$.
    Hence, the arrow~$(y', \beta'(y')) = (y', \beta(y))$ lies within the fan of~$\dom(\Q'.\phi)$.
    In both cases, the~$\Q'$-shift of~$(y', \beta'(y'))$ is strictly less than~$\resdist{\P}$.
    Hence, the~$\Q'$-shift of~$\beta'$ is strictly less than~$\resdist{\P}$.

    It remains to show that~$(\alpha',\beta')$ is an interleaving.
    Observe that for all~$y' \in T_2$, the point~$\beta'(y')$ is an ancestor of~$\beta(y')$. 
    Let~$x' \in T_1$.
    If~$\alpha'(x') = \alpha(x')$, then we obtain~$\beta'(\alpha'(x')) = \beta'(\alpha(x')) \succeq \beta(\alpha(x'))$.
    As~$\I$ is an interleaving, it follows that~$\beta'(\alpha'(x'))$ is an ancestor of~$x'$.
    Otherwise, if~$\alpha'(x') \neq \alpha(x')$, then, by construction, the point~$\alpha'(x')$ is an ancestor of~$y^*$.
    In particular, this implies that~$\beta'(\alpha'(x'))$ is an ancestor of~$\beta(y)$.
    Since~$\I$ is an optimal~$\Q$-extension, we know that~$\alpha(x) = y$.
    Moreover, it must be the case that~$\beta(y)$ is an ancestor of~$x$.
    It follows that~$\beta'(\alpha'(x'))$ is also an ancestor of~$x$, and thus also of~$x'$.    
    Finally, let~$y' \in T_2$.
    Observe that for all~$x' \in T_1$, the point~$\alpha'(x')$ lies on the same edge as~$\alpha(x')$.
    If~$\beta'(y') = \beta'(y)$, then the point~$\alpha'(\beta'(y')) = \alpha'(\beta(y'))$ lies on the same edge as~$\alpha(\beta(y')$.
    As~$\I$ is an interleaving, it follows that~$\alpha'(\beta'(y'))$ is an ancestor of~$y'$.
    Otherwise,~$\beta'(y')$ is an ancestor of~$\beta(y')$.
    Since~$\alpha'$ preserves ancestors, we get that~$\alpha'(\beta'(y'))$ is an ancestor of~$\alpha'(\beta(y'))$.
    The point~$\alpha'(\beta(y'))$ lies on the same edge as~$\alpha(\beta(y'))$, so again, as~$\I$ is an interleaving, it follows that~$\alpha'(\beta(y'))$ is an ancestor of~$y'$.
    So, the point~$\alpha'(\beta'(y'))$ is also an ancestor of~$y'$.


    \proofsubparagraph{Concluding the proof.}
    In all cases, we obtain an interleaving~$\I'$ that has~$\Q'$-shift strictly less than~$\resdist{\P}$.
    In other words,~$\Q'$ is a~$\P$-augmentation, and we reach a contradiction.
\end{proof}

\subsection{Constructing a Locally Correct Interleaving}\label{sec:locally-correct-partial-merge-tree-matchings}
We show that a locally correct interleaving always exists by incrementally constructing one.
Specifically, we iteratively replace~$\P$ with a finite minimal~$\P$-augmentation, and we show that after a finite number of iterations the~$\P$-residual interleaving distance becomes zero, at which point we stop iterating.
We conclude by showing that for the resulting partial interleaving~$\P$, any optimal~$\P$-extension is locally correct.
For this, we maintain that~$\P$ remains ``partially'' \emph{locally correct} throughout the entire construction.
Moreover, we show that after each step of the construction, $\P$ ``specifies'' at least one additional vertex of the input merge trees.
Since each vertex can be specified at most once, it follows that the construction terminates.

\subparagraph{Minimal augmentations.}
Formally, let~$\mathcal{P}$ be a partial interleaving and assume~$\resdist{\P} > 0$.
Let~$\Q$ be a finite minimal~$\P$-augmentation.
By definition,~$\resdist{\Q} < \resdist{\P}$.
If for all restrictions~$\R$ of~$\P$ it holds that~$\resdist{\R} \ge \shift{\R}(\P)$, we say~$\P$ is \emph{locally correct}.
Note that for a complete interleaving this definition is equivalent to Definition~\ref{def:locally-correct} (this follows from Observation~\ref{obs:locally-correct}).
\begin{lemma}\label{lem:dominant-locally-correct}
Let~$\P$ be a finite interleaving and suppose~$\resdist{\P} > 0$.
Let~$\Q$ be a finite minimal~$\P$-augmentation.
If~$\P$ is dominant and locally correct, then so is~$\Q$.
\end{lemma}
\begin{proof}
    We first argue that~$\Q$ is dominant.
    Let~$a$ be an arrow of~$\Q$.
    If~$a$ is also an arrow of~$\P$, then from the fact that~$\P$ is dominant we know that~$\shift{}(a) > \resdist{\P} > \resdist{\Q}$.
    Otherwise, $a$ must be an arrow of the bottleneck~$\Q \setminus \P$.
    By Lemma~\ref{lem:only-critical-pairs}, it then directly follows that~$\shift{}(a) = \resdist{\P} > \resdist{\Q}$.

    Next, we argue that~$\Q$ is locally correct.
    Consider a restriction~$\R$ of~$\Q$; we need to show that~$\resdist{\R} \ge \shift{\R}(\Q)$.
    If~$\R = \Q$, we immediately obtain~$\shift{\R}(\Q) = 0 \le \resdist{\R}$.
    Otherwise, let~$\R'$ be the restriction of~$\R$ to the domains of~$\P.\phi$ and~$\P.\psi$.
    Since~$\Q$ uses~$\P$, it follows that~$\R'$ uses~$\P$ and is hence a restriction of~$\P$.
    If~$\R' = \P$, then by definition of an augmentation we know that the~$\R'$-shift of any arrow of~$\Q$ is~$\resdist{\P}$.
    We obtain~$\shift{\R'}(\Q) = \resdist{\P}$.
    Moreover, since~$\R \neq \Q$, the~$\R$-shift of at least one arrow of~$\Q \setminus \P$ is non-zero.
    We obtain~$\shift{\R}(\Q) = \shift{\R'}(\Q) = \resdist{\P}$.
    Since~$\Q$ is a minimal augmentation and extends~$\R$, we know that~$\R$ cannot be an augmentation of~$\P$.
    It follows that~$\resdist{\R} \ge \resdist{\P}$.
    Putting it together, we obtain~$\resdist{\R} = \shift{\R}(\Q)$.
    
    Lastly, if~$\R' \neq \P$, then since~$\P$ is locally correct and~$\R'$ is a restriction of~$\P$, we have
    \begin{equation}\label{eq:helper}
        \resdist{\R'} \ge \shift{\R'}(\P).
    \end{equation}
    We first show that~$\shift{\R'}(\P) \ge \shift{\R}(\Q)$.
    First, none of the arrows of~$\R \setminus \R'$ are arrows of~$\P$, so~$\shift{\R'}(\P) = \shift{\R}(\P)$.
    Next, the~$\R$-residual shift of~$\Q$ is the maximum of the~$\R$-residual shifts of~$\P$ and~$\Q \setminus \P$.
    Since~$\P$ is dominant, we know that the shift of any arrow of~$\P$ is strictly greater than~$\resdist{\P}$.
    As~$\P \setminus \R$ is not empty, it follows that~$\shift{\R}(\P) > \resdist{\P}$.
    On the other hand, by definition of an augmentation, we know that the shift of any arrow of~$\Q \setminus \P$ is at most~$\resdist{\P}$.
    So, we get~$\shift{\R}(\Q) = \max(\shift{\R}(\P), \shift{\R}(\Q \setminus \P)) = \shift{\R}(\P)$.
    Combining, we obtain~$\shift{\R'}(\P) \ge \shift{\R}(\Q)$.

    Next, we argue that~$\resdist{\R} \ge \resdist{\R'}$.
    By construction,~$\R$ extends~$\R'$.
    As observed before, we have~$\shift{\R'}(\P) = \shift{\R}(\P) > \resdist{\P}$.
    Combining with Equation~\eqref{eq:helper}, we obtain~$\resdist{\R'} \ge \shift{\R'}(\P) > \resdist{\P}$.
    Let~$a$ be an arrow of~$\R \setminus \R'$.
    Then~$a$ is arrow of the bottleneck~$\Q \setminus \P$, so the shift of~$a$ is exactly~$\resdist{\P}$.
    In other words,~$\shift{\R'}(a) = \resdist{\P} < \resdist{\R'}$.
    This means that we can apply Corollary~\ref{cor:helper-lemma}, for the finite interleaving~$\R'$ and the~$\R'$-extension~$\R$, to obtain~$\resdist{\R} \ge \resdist{\R'}$.
    
    Thus, $
        \resdist{\R} \ge \resdist{\R'} \ge \shift{\R'}(\P) = \shift{\R}(\Q),
    $
    which concludes the proof.
\end{proof}

\noindent
Recall that for each arrow~$(x, y)$, with~$x \in T_1$ and~$y \in T_2$, of~$\P$, both~$x$ and~$y$ are critical points.
However, only the point~$x$ has a ``specified'' target in the other tree.
We say a point~$x \in T_1$ or~$y \in T_2$ is \emph{specified} by~$\P$, if~$x \in \dom(\P.\phi)$ or~$y \in \dom(\P.\psi)$ respectively (see Figure~\ref{fig:specified-points}).
We argue that if~$\P$ specifies all critical points that are not vertices, i.e.,~$C_1[\P] \setminus V(T_1) \subseteq \dom(\P.\phi)$ and~$C_2[\P] \setminus V(T_2) \subseteq \dom(\P.\psi)$, then so does any minimal augmentation.

\begin{figure}
    \centering
    \includegraphics{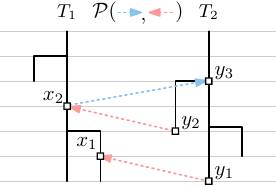}
    \caption{The points~$x_1$, $x_2$, $y_1$, $y_2$ and~$y_3$ are all critical, but only~$x_2$, $y_1$ and~$y_2$ are specified.}
    \label{fig:specified-points}
\end{figure}

\begin{lemma}\label{lem:specifies-subdivided}
    Let~$\P$ be a finite interleaving and suppose~$\resdist{\P} > 0$.
    Let~$\Q$ be a finite minimal~$\P$-augmentation.
    If~$\P$ is dominant and specifies all critical points that are not vertices, then so does~$\Q$.
    Moreover,~$\Q$ specifies at least one vertex that~$\P$ does not specify.
\end{lemma}
\begin{proof}
	By Lemma~\ref{lem:only-critical-pairs}, we know that the bottleneck~$\Q \setminus \P$ uses only~$\P$-critical pairs.
	Consider such a pair~$b$.
	If~$b$ is an arrow critical pair, then without loss of generality assume~$b=(x, y)$ for a point~$x \in T_1$ and a point~$y \in T_2$; the case~$b=(y, x)$ is symmetric. 
    By definition of an arrow critical pair, both~$x$ and~$y$ are~$\P$-critical points.
    So the pair~$b$ does not create any additional~$\Q$-critical points.
    As~$\P$ already specifies all critical points that are not vertices and~$\Q$ uses~$\P$, we know that~$x$ must be a vertex.
    Since~$\P$ is dominant, the shift of any arrow of~$\P$ is strictly greater than~$\resdist{\P}$.
    Moreover, the shift of any arrow of~$\Q \setminus \P$ is exactly~$\resdist{\P}$.
    Hence, the vertex~$x$ cannot be specified by~$\P$: at least one new vertex is specified.
    
	If~$b$ is a zigzag critical pair, then without loss of generality assume~$b = (x_1, x_2)$ for two points~$x_1$ and~$x_2$ of~$T_1$; if it is a pair of points of~$T_2$ the argument is analogous.
	By definition of a zigzag critical pair, we know that both~$x_1 \in V(T_1)$ and~$x_2 \in V(T_1)$.
    Moreover, there is a point~$y \in T_2$ such that~$Q.\phi(x_1) = y$ and~$Q.\psi(y) =  x_2$.
    So,~$x_1$ is a vertex that is specified by~$\Q$ but not by~$\P$.
    Lastly,~$b$ creates a single~$\Q$-critical point~$y$, which is immediately specified.
\end{proof}

\subparagraph{Finishing the proof of Theorem~\ref{thm:existence-locally-correct}.}
Let~$\P$ be the empty interleaving.
Trivially,~$\P$ is dominant and locally correct.
Moreover, there are no (unspecified) critical points that are not vertices.
Now, we iteratively replace~$\P$ with a finite minimal~$\P$-augmentation until the~$\P$-residual interleaving distance becomes zero.
By applying Lemmas~\ref{lem:dominant-locally-correct} and~\ref{lem:specifies-subdivided} inductively, we maintain that~$\P$ is dominant and locally correct, and specifies all critical points that are not vertices.
Additionally, by Lemma~\ref{lem:specifies-subdivided}, each iteration specifies at least one additional vertex.
So after $i$ iterations, at least $i$ vertices are specified, and the total number of iterations is at most~$|V(T_1)| + |V(T_2)|$.
Afterwards, the residual interleaving distance is zero.
Lemma~\ref{lem:optimal-extension-locally-correct} guarantees that a locally correct interleaving exists; its proof is similar to that of Lemma~\ref{lem:dominant-locally-correct}.

\begin{restatable}{lemma}{thefinallemma}
\label{lem:optimal-extension-locally-correct}
    Let~$\P$ be a locally correct partial interleaving with~$\resdist{\P} = 0$.
    Any optimal~$\P$-extension is locally correct.
\end{restatable}
\begin{proof}
    Fix a restriction~$\R$ of~$\I$.
    By Observation~\ref{obs:locally-correct} it suffices to show that~$\resdist{\R} \ge \shift{\R}(\I)$.
    If~$\R$ uses~$\P$, then we immediately have~$\resdist{\R} = 0 = \shift{\R}(\I)$.
    Otherwise, let~$\R'$ be the restriction of~$\R$ to the domains of~$\P.\phi$ and~$\P.\psi$.
    Observe that~$\R'$ is also a restriction of~$\P$.
    Since~$\P$ is locally correct, we get
    \begin{equation}\label{eq:lc-1}
    \resdist{\R'} \ge \shift{\R'}(\P).
    \end{equation}

    Since~$\R$ does not use~$\P$, neither does~$\R'$.
    This means that there is at least one arrow of~$\P$ that is not in~$\R'$.
    We know that~$\P$ is dominant, so the shift of any such arrow is strictly larger than~$0$.
    As a result, we have~$\shift{\R'}(\P) > 0$.
    
    Let~$a = (x, y)$ be an arrow of~$\R \setminus \R'$.
    Without loss of generality we assume~$x \in T_1$ and~$y \in T_2$; the symmetric case is analogous.
    Since~$\I$ is an optimal~$\P$-extension, we have~$\shift{\P}(a) = 0$.
    In other words, we have either~$\shift{}(a) = 0$ or~$a \in F[A_\P]$.
    We argue that in both cases we have~$\shift{\R'}(a) < \shift{\R'}(\P)$.
    If~$\shift{}(a) = 0$, then we directly obtain~$\shift{\R'}(a) = 0 < \shift{\R'}(\P)$.
    Otherwise, if~$a \in F[A_\P]$, then there exists an arrow~$a' = (x', y) \in A_\P$ such that~$x'$ is a strict descendant of~$x$.
    This means that the shift of~$a'$ is strictly greater than the shift of~$a$.
    We consider two cases.
    First, if~$a' \in A_{\R'}$ then~$a$ is in the fan of~$A_{\R'}$ and we obtain~$\shift{\R'}(a) = 0 < \shift{\R'}(\P)$.
    Otherwise, if~$a' \notin A_{\R'}$, then~$\shift{\R'}(a') = \shift{}(a')$.
    Hence, we obtain~$\shift{\R'}(\P) \ge \shift{}(a') > \shift{}(a) \ge \shift{\R'}(a)$.
    In all cases, we obtain~$\shift{R'}(a) < \shift{\R'}(\P)$.
    Combining with Equation~\eqref{eq:lc-1}, we get~$\shift{\R'}(a) < \resdist{\R'}$.
    By Corollary~\ref{cor:helper-lemma}, it then follows that~$\resdist{\R} \ge \resdist{\R'}$.

    It remains to show that~$\shift{\R'}(\P) \ge \shift{\R}(\I)$.
    Since any arrow of~$\R \setminus \R'$ is not an arrow of~$\P$, we have~$\shift{\R'}(\P) = \shift{\R}(\P)$.
    Consider an arrow~$a = (x, y)$ of~$\I$.
    Again, without loss of generality assume~$x \in T_1$ and~$y \in T_2$; the symmetric case is analogous.
    The~$\P$-shift of any arrow~$a$ is~$0$, so either~$\shift{}(a) = 0$, or $a \in F[A_\P]$.
    In the first case, we get~$\shift{\R}(a) = 0 < \shift{\R}(\P)$.
    In the second case, there must be an arrow~$a' = (x',y)$ of~$\P$ with~$x' \preceq x$.
    It follows that the shift of~$a'$ is at least as great as the shift of~$a$, and thus~$\shift{\R}(a) \le \shift{\R}(a') \le \shift{\R}(\P)$.
    Since this holds for all arrows of~$\I$, we get~$\shift{\R}(I) \le \shift{\R}(\P) \le \shift{\R'}(\P)$.
    This concludes the proof.
\end{proof}

\subsection{Computing a Locally Correct Interleaving}\label{sec:computing}
We sketch how the proof of Theorem~\ref{thm:existence-locally-correct} can be translated into an incremental algorithm for computing a locally correct interleaving.
To compute the residual interleaving distance, we modify the dynamic program by Touli and Wang~\cite{touli2022fpt} to also incorporate constraints induced by partial interleavings.
We use the output to extract a finite set of critical pairs, that corresponds to an augmentation.
We then iteratively reduce the set of critical pairs by repeatedly calling the modified dynamic program until we obtain a (finite) minimal augmentation.
This minimal augmentation, in turn, reduces the residual interleaving distance; we iterate until it has dropped to zero.
To obtain a complete locally correct interleaving, we select suitable entries from the table filled by the modified dynamic program.

\section{Conclusion}\label{sec:conclusion}
We introduced locally correct interleavings between merge trees, which are optimal interleavings that are ``tight''.
For this, we introduced the residual interleaving distance: a generalized version of the interleaving distance that can accommodate constraints on the interleavings.
We presented a constructive proof that there always exists a locally correct interleaving.

We plan to further investigate the complexity of computing a locally correct interleaving.
In particular, we will study the effect of partial interleavings on the running time of the modified dynamic program.
Furthermore, we will explore the effects of replacing the dynamic program with the approximation algorithm~\cite{agarwal2018computing} or one of the variants or heuristics.
Moreover, it would be interesting to study other (local) criteria of interleavings.
Since locally correct interleavings are not necessarily unique, we could minimize their \emph{total} shift, for a suitable definition of total.
Furthermore, we could also consider a \emph{lexicographic} optimization, inspired by lexicographic Fréchet matchings~\cite{rote2014lexicographic}.
Lastly, we want to explore in which way locally correct interleavings can be translated to meaningful matchings between terrains.

\bibliography{references.bib}
\end{document}